\documentclass[11pt]{article}
\usepackage[toc,page]{appendix}
\usepackage{graphicx}
\usepackage{floatrow}   
\usepackage{amsmath}
\RequirePackage{amssymb}
\RequirePackage{amsthm}
\usepackage[scr=boondoxo]{mathalfa}
\usepackage{caption}
\usepackage{amsmath,scalefnt}
\usepackage{subfig}
\usepackage{caption}
\usepackage{footnote}
\usepackage{fmtcount}
\usepackage{color}
\usepackage{multicol}
\usepackage{booktabs}
\usepackage{multirow}
\usepackage{float}
\usepackage{array}
\usepackage{bm}
\usepackage{rotating}
\usepackage{enumitem}

\usepackage{umoline}

\usepackage[normalem]{ulem}

\usepackage{lipsum}
\usepackage{comment}
\usepackage{threeparttable}
\usepackage{lscape}
\usepackage{umoline}
\usepackage{threeparttable,graphicx}
\usepackage{chemarr}

\renewcommand{\aa}{a}

\definecolor{gr}{gray}{0.8}

\newcommand\redsout{\bgroup\markoverwith{\textcolor{red}{\rule[0.5ex]{2pt}{1.0pt}}}\ULon}
\newcommand\bluesout{\bgroup\markoverwith{\textcolor{blue}{\rule[0.5ex]{2pt}{1.0pt}}}\ULon}

	\theoremstyle{plain}

	\newtheorem{lemma}{Lemma}

	 \newtheorem{assumption}{Assumption}
	 \newtheorem{proof of proposition}{Proof of Proposition}
	\newtheorem{proposition}{Proposition}

	\theoremstyle{definition}
  \newtheorem{remark}{Remark}

\renewcommand{\r}{{\mathbb R}}

\newcommand{\beqn}{\begin{eqnarray*}}
\newcommand{\eeqn}{\end{eqnarray*}}

\newcommand{\be}[1]{\begin{equation}\label{#1}}
\newcommand{\ee}{\end{equation}}
\newcommand{\bi}{\begin{itemize}}
\newcommand{\ei}{\end{itemize}}
\newcommand{\ben}{\begin{enumerate}}
\newcommand{\een}{\end{enumerate}}

\title{A Mathematical Description of \\Bacterial Chemotaxis  in Response to Two Stimuli}
\author{Jeungeun Park 
\thanks{Department of Mathematical Sciences, University of Cincinnati, OH, USA.
	{\tt\small  park4ju@ucmail.uc.edu}}
\and Zahra Aminzare
\thanks{Department of Mathematics,  University of Iowa, IA, USA.
	{\tt\small zahra-aminzare@uiowa.edu}}
	}
\date{}
\begin{document}
\maketitle

\begin{abstract}
\noindent
Bacteria are often exposed to multiple stimuli in complex environments, and their efficient chemotactic decisions are critical to survive and grow in their native environments. 
Bacterial responses to the environmental stimuli depend on the ratio of their corresponding chemoreceptors. 
By incorporating the signaling machinery of individual cells, 
we analyze the collective motion of 
a population of \emph{Escherichia coli} bacteria in response to two stimuli, mainly serine and methyl-aspartate
(MeAsp), in a one-dimensional and a two-dimensional environment, which is inspired by experimental results in Y. Kalinin \emph{et al.}, J. Bacteriol. 192(7):1796--1800, 2010.
Under suitable conditions, we show that if the ratio of the main chemoreceptors of individual cells, namely Tar/Tsr is less than a specific threshold, the bacteria move to the gradient of serine, and if the ratio is greater than the threshold, the group of bacteria move toward the gradient of MeAsp.
Finally, we examine the theory with Monte-Carlo agent-based simulations, and verify that our results qualitatively agree well with the experimental results in Y. Kalinin \emph{et al.} (2010).
\end{abstract}
\textbf{Key words.} 
Chemotaxis, 
Multi-scale dynamics, 
Population dynamics,
Intracellular decision making, 
Fokker-Planck equations,
Advection-diffusion equations, Monte-Carlo simulations.

\textbf{Mathematics Subject Classification (2020). } 35Q92, 58J55, 60J75, 92B05, 92C17, 92D25

%%%Introduction 
\section{Introduction} \label{introduction}

The preferred movement of a bacterium along the gradient of chemical substances, the so-called chemotaxis,  includes a directed movement (run) and a relatively short random turning (tumble). 
See e.g., \cite{berg1972chemotaxis} and \cite{macnab1972gradient} for 
{\em Escherichia coli} (\emph{E.~coli}) and
{\em Salmonella typhimurium} chemotaxis.
Each bacterium carries an internal state which may be modeled by a system of ordinary differential equations. In the presence of a stimulus in the environment, each cell changes its direction at random, with a tumbling rate which depends on the internal state, biasing moves toward more favorable environments or away from noxious substances.

In natural environments, bacteria are often exposed to multiple chemical stimuli. To navigate toward a favorable environment, they choose their directions of movement based on environmental perception, individual preferences, and interaction with others. Also, each individual's decision characterizes the behavior of a group of bacteria. Thus, understanding how bacterium chooses between multiple stimuli is essential to study bacterial chemotaxis at the population level.

In the case of \emph{E.~coli}, chemical signals are often detected via five main chemoreceptors, namely Tar, Tsr, Tap, Trg, and Aer \cite{vladimirov2009chemotaxis}.
In \cite{kalinin2010responses}, where responses of \emph{E.~coli} to two chemoattractant signals are demonstrated, it is shown that 
the expression levels of the most abundant receptors, Tar and Tsr, are determined by the bacterial density in a batch-mode culture within the growth phase; in turn, the ratio of these receptors differentiates their chemical preferences.

Inspired by the experimental results of  \cite{kalinin2010responses}, our goal of this work is to incorporate the bacterial decision-making process into a mathematical model and investigate  the corresponding collective behavior observed in  \cite{kalinin2010responses}. 
To this end, we consider a population of bacteria in a one-dimensional and a two-dimensional spatial domain occupied by two stimuli that their temporal rates are assumed to be zero. 
First, we employ a Fokker-Planck type master equation (also known as balance equation \cite{alt1980biased}) 
to describe the bacterial chemotaxis.  This (microscopic) model enables us to incorporate the internal dynamics of \emph{E.~coli} representing the chemotaxis signaling pathway \cite{tu2008modeling,edgington2018mathematical}. Then, we describe the \emph{E.~coli} population dynamics by a (macroscopic) advection-diffusion equation, which is analogous to the classic Keller-Segel model \cite{keller1971model}, and can be derived from the microscopic model by the tools developed in \cite{erban2004individual}.

Mathematical modeling aiming to understand the behavior of bacteria population in response to external signals has been extensively studied (see \cite{tindall2008overview} for a review on multi-scaling model approaches for chemotaxi). 
In \cite{erban2004individual, erban2005signal}, the authors studied \emph{E.~coli} chemotaxis in response to a single stimulus in a one-dimensional and an arbitrary dimensional
space, respectively. 
These studies were generalized in \cite{xue2009multiscale}  to 
multiple space- and time-dependent signals by applying a  general type of receptor based-response laws
\cite{othmer1997aggregation, painter2000development}. 
These works considered a {\it toy} model for the internal dynamics. 
In \cite{aminzare2013remarks, menolascina2017logarithmic}, the authors allow arbitrary one-dimensional internal dynamics in response to a time-independent 
signal and more realistic models for \emph{E.~coli} internal dynamics given in  \cite{tu2008modeling, kalinin2009logarithmic}. The theory was further generalized to higher dimensional  space and multiple signals in \cite{xue2015macroscopic, xue2016moment}. 
The authors in \cite{hu2014behaviors} 
incorporated \emph{E.~coli} signalling pathway from \cite{tu2008modeling} into a one-dimensional macroscopic equation in order to understand various taxis behaviors in \cite{salman2007concentration, demir2011effects, yang2012opposite}. The macroscopic model was also validated by comparing with available experimental data that show the ratio of Tar and Tsr affects bacterial thermotaxis and pH taxis.

Our contributions towards understanding the dynamics of a population of bacteria in response to two stimuli are as follows.
First, we incorporate a relatively \textit{general} class of  one-dimensional internal dynamics into a one- and a two-dimensional microscopic equation
from which derives a macroscopic equation.  
Second, we use the macroscopic model for a population of \emph{E.~coli} with a mechanistically realistic, while a mathematically tractable, model of internal dynamics and analyze the response of \emph{E.~coli} to two stimuli in a one- and a two-dimensional environment. By analyzing the steady state solution of the macroscopic equation, we further show that there is a critical ratio of receptors that determines bacterial movement toward their favored chemical.
Finally, we demonstrate some Monte-Carlo agent-based simulations for different types of stimuli and compare them with numerical solutions of the model. We also explain that  the Monte-Carlo simulations results
agree well with the experimental results of \cite{kalinin2010responses}.

The remainder of the paper is organized as follows. 
In Section \ref{Microscopic}, we first review the internal dynamics of \emph{E. coli} which describe how the cells can produce runs and tumbles. 
Then, given a general internal dynamics of bacteria, we introduce a (forward) Fokker-Planck equation which describes the dynamics of a probability distribution of a population of bacteria.
In Section \ref{AD:1dim} (respectively, Section \ref{AD:2dim}), we first derive a
one-dimensional (respectively, two-dimensional) advection-diffusion equation which approximates the 
Fokker-Planck equation with a general internal dynamics. Then, we focus on a population of \emph{E. coli} with a specific internal dynamics. 
Also, a bifurcation parameter and its value of bacterial chemical preferences are identified. It is further verified by comparing the solutions to the advection-diffusion equations with those of Monte-Carlo agent-based simulations in
 Section \ref{simulation:1dim} (respectively, Section \ref{simulation:2dim}) for different combinations of stimuli. 
In Section \ref{Discussion}, we conclude with a brief summary and discussion of future directions.
In Appendix \ref{Appendix}, we summarize the models with parameter values that we use for the internal dynamics in Section \ref{Microscopic} and for the derivation of the macroscopic equation in Sections \ref{AD:1dim} and \ref{AD:2dim}. The appendix also  provides a brief description of the Monte-Carlo agent-based simulation and an overview of our numerical simulations with input data.

%%% Section 2
\section{Microscopic behavior of a population of \emph{E.\ coli} bacteria}\label{Microscopic}

We briefly review the internal dynamics of \emph{E.\ coli} which transfer a signal of the environment into a motor rotation for a run or a tumble (see \cite{tu2008modeling, kalinin2009logarithmic, jiang2010quantitative} for more details).  Then, following \cite{erban2004individual, aminzare2013remarks, Othmer1988ModelsOD},  we derive a probabilistic equation which describes  \textit{microscopic} dynamics of a population of bacteria with a given internal dynamics.
Later, in the following section, we use the microscopic equation to derive a \textit{macroscopic} 
equation which approximates the dynamics of a population of bacteria by integrating   the internal dynamics of all the bacteria. 

%****************
\subsection{The internal dynamics of \emph{E.\ coli}: A brief review}
\label{internal:review}

\emph{E.~coli} bacteria use four to six helical flagella that are connected to rotary motors in their cell wall to swim. Their swimming patterns are characterized as a random walk, consisting of long runs ($\sim 1$ sec) and short tumbles ($\sim 0.1$ sec). When a cell senses an increasing of external attractant gradient, the run length is extended \cite{berg1972chemotaxis,berg1990chemotaxis}.
The receptors on the membrane of the cells, which receive the signals, 
and the flagella motors, which produce runs and tumbles, are connected by a signaling pathway within the cell, as shown in Figure~\ref{integral:feedback}(left),  \cite{wadhams2004making}. 
Each receptor is linked to a histidine kinase CheA, through a linker protein CheW. 

In the absence of an attractant gradient, CheA autophosphorylates and  produces CheA-P.  
Phosphoryl group of CheA-P transfers to either CheY or CheB.
 Phosphorylated CheY (denoted by CheY-P) increases the probability of tumbles by rotating the motor clockwise \cite{welch1993phosphorylation,bren1996signal, lipkow2005simulated}. 
CheZ accelerates the dephosphorylation of CheY-P, which quickly modulates the motion of flagella \cite{lipkow2006changing}.
 
In the presence of an attractant gradient, a ligand binds to a receptor and inhibits the activity of CheA, followed by decreasing the CheY-P and CheB-P levels. The reduction in CheY-P levels lengthens the run with a counter-clockwise motor rotation.

To respond to further changes in the concentration of a gradient, 
CheR and CheB-P mediate adaptation. 
On the one hand, CheR methylates the receptors and hence enhances CheA activity \cite{springer1977identification}. 
On the other hand, CheB-P demethylates the receptors and consequently inhibits the activity of CheA \cite{stock1978protein}. Therefore, when an attractant gradient is sensed, the CheA-P level, and thus the CheB-P level decrease. 
While the CheB-P level decreases, the receptors are methylated by CheR, and they return to their pre-stimulus state, followed by the pre-stimulus values of CheA activity, CheA-P and CheY-P levels, and motor bias. This process is called an adaptation of methylation.

The intracellular chemotaxis signaling pathway, which contains three main phosphorylation groups and the receptor methylation level, can be mathematically modeled by four coupled ordinary differential equations (ODEs) that consist of three biochemical equations for CheA-P, CheB-P, and CheY-P, and one equation for the  methylation level of receptors. 
However, the phosphorylation processes and the methylation process occur at different time scales, and one can reduce the 4-dimensional system into a 3-, 2-, or even a 1-dimensional system. In \cite{edgington2018mathematical}, the authors explained these reductions in detail. 

It is known that the adaptation process of methylation is much slower
than the other dynamics in the signaling pathway  \cite{erban2005signal, bray1995computer,terwilliger1986kinetics,simms1987purification}. 
Therefore, assuming quasi-equilibrium approximations for CheA-P, CheB-P, and CheY-P,
we consider a one-dimensional reduction model for the methylation level of receptors, as developed in \cite{tu2008modeling}.  

Consider the following input-output dynamics for the chemotaxis signaling pathway, as shown in Figure~\ref{integral:feedback}(right). 
The ligand concentration, denoted by $S$, and the tumbling rate, denoted by $\lambda$, represent the input and the output, respectively. 
As explained above, binding the ligand to the receptor inhibits the activity of CheA, denoted by $\aa$. On the other hand, the methyl group (denoted by $m$) in the receptors enhances the activity of $\aa$. Therefore, $\aa= G(S,m)$ can be described as an increasing function of $m$ and a decreasing function of $S$. 

\medskip

%******Fig 1******
\begin{figure}[ht]
\floatbox[{\capbeside\thisfloatsetup{capbesideposition={right,top},capbesidewidth=6cm}}]{figure}[\FBwidth]
{\includegraphics[width=8cm]{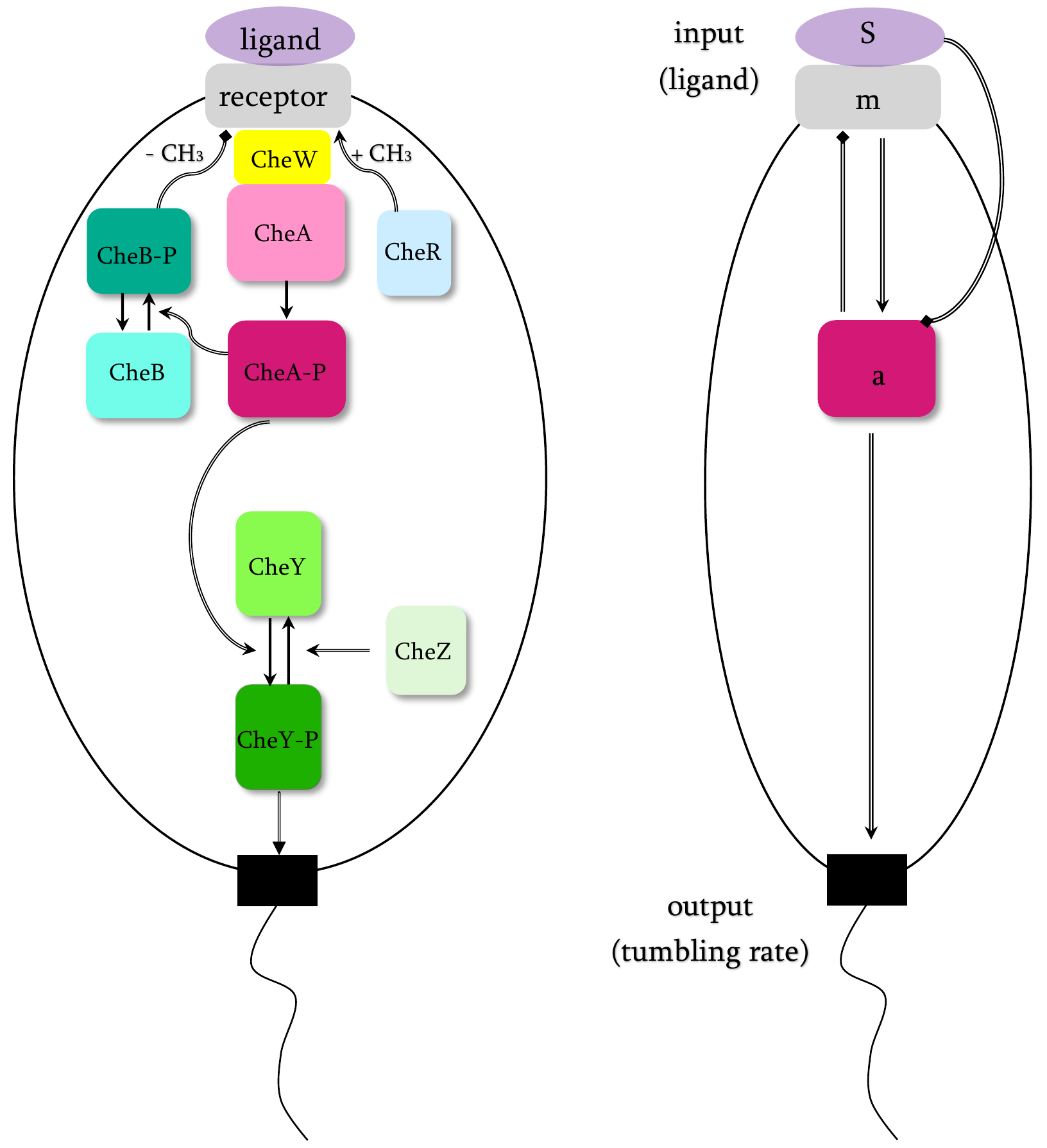}}
{ \caption{\small{{Left:{
\emph{E.~coli} signaling pathway}.
Binding ligands to receptors, the signal is transduced to the flagellar motor via six cytoplasmic chemotaxis proteins. Right: An input-output representation of \emph{E.~coli} signaling pathway. The internal signaling pathway shown in left  is reduced to the interaction between the methylation level $m$ and the kinase activity $\aa$. This interaction, which depends on ligand concentration $S$ (input), controls the motor rotation by changing the tumbling rate (output). See Section \ref{internal:review} for detailed description.}}
}\label{integral:feedback}}
\end{figure}
As described earlier, the kinase activity of CheA enhances the CheB-P level, and CheB-P reduces the methylation level of the receptors. Consequently, the kinase activity $\aa$ reduces the methylation level $m$, indirectly. So, the dynamics of $m$ can be described by ${dm}/{dt} = F(\aa)$, where $F$ is a decreasing function of $\aa$.

Several models for methylation dynamics ($F$) and kinase activity function ($G$) have been developed. See e.g., \cite{tu2008modeling,  edgington2018mathematical, vladimirov2008dependence, clausznitzer2010chemotactic}. For ease of calculation, we choose models for $F$ and $G$ as described in \eqref{output:a} and \eqref{methylation:m} below. Despite of the simplicity, the models capture the essential features such as receptor cooperativity, methylation on kinase activity and adaptation for \emph{E.~coli} signaling pathway, and they were verified by existing experiments. More details are discussed below.

Note that the tumbling rate $\lambda$ is controlled by the level of CheY-P, which is affected by the kinase activity. Therefore, $\lambda$ can be modeled by an increasing function of $\aa$, as described in \eqref{tumble_rate} below. 

Following the experimental set up in Kalinin {\em et al.} \cite{kalinin2010responses}, we consider two stimuli: $S_1$ and $S_2$, which, respectively, stand for methyl-aspartate (MeAsp) and serine, and can be sensed by chemoreceptors Tar and Tsr. 
Furthermore, since the experiments in \cite{kalinin2010responses} are designed 
to keep the external signals $S_{1}$ and $S_{2}$ constant in time, we assume that $S_1$ and $S_2$ only depend on the spatial variable $\mathbf{x}$ and are independent of time $t$: $S_1= S_1(\mathbf{x})$ and $S_2=S_2(\mathbf{x}).$

Following \cite{hu2014behaviors,mello2005allosteric, neumann2010differences}, we let a heterogeneous Monod-Wyman-Chageux (MWC) model \cite{monod1965nature} describe 
the kinase activity $\aa$:
\begin{equation}\label{output:a}
G(S_{1}, S_{2}, m) = \dfrac{1}{1+\eta_0(m) \eta_1(S_1) \eta_2(S_2)},
 \end{equation}
 where $\eta_{0}(m)\eta_{1}(S_{1})\eta_{2}(S_{2})$ is derived from the total free energy difference between the active and inactive states. 
 According to \cite{ tu2008modeling, jiang2010quantitative, sourjik2002receptor, shimizu2006monitoring, mello2007effects}, the methylation-dependent free energy gives 
 \[\eta_0(m) = \exp(N\alpha (m_0-m)),\] 
 where $N$ is the number of the responding receptor dimers in the cluster, and $\alpha$ and $m_{0}$ denote the free-energy per added methylation group and a reference methylation level, respectively.
 The ligand-depdent free-energy obtains 
\[
\eta_i(S_i) = \left(\dfrac{1+{S_i}/{K_I^i}}{1+{S_i}/{K_A^i}}\right)^{N r_i}, 
\]
where $K_I^i$ and $K_A^i$
are the dissociation constants of the corresponding ligand ($i=1$ for MeAsp, $i=2$ for serine) to the inactive and the active receptor ($i=1$ for Tar, $i=2$ for Tsr).
The constant parameters 
$r_1$ and $r_2$ 
are the fraction of receptors Tar and Tsr in the receptor cluster, respectively. We assume that $r_{1} + r_{2} =1$ and $r_1N$ and $r_2N$ are the number of the receptors binding to the corresponding ligand. 

 The \textit{average} methylation level, $m$,  of receptors evolves slowly and can be described by the following equation \cite{tu2008modeling, jiang2010quantitative}:
 \begin{equation}\label{methylation:m}
 \dfrac{dm}{dt} = F(a) = \dfrac{a_0 - a}{\tau_a},
 \end{equation}
where $\tau_a\gg1$ is the time scale and $a_0$ is a constant which represents the adaptation level of $\aa$, i.e., when $\aa>a_0$, ${dm}/{dt}<0$ and hence $m$ and consequently $\aa$ decrease.  When $\aa<a_0$,  ${dm}/{dt}>0$ and hence $m$ and consequently $\aa$ increase. 

It is more convenient to use $\aa$ as a state variable instead of the methylation level $m$.  
Taking time derivative of $a$ gives:
\begin{align}\label{dynamics:a}
 \dfrac{da}{dt} \;=\; \dfrac{\partial a}{\partial m}\;\dfrac{d m}{d t} +  \dfrac{\partial a}{\partial S_1}\;\nabla_{\mathbf{x}}{S_1}\cdot\dfrac{d\mathbf{x}}{dt} +  \dfrac{\partial a}{\partial S_2}\;\nabla_{\mathbf{x}}{S_2}\cdot\dfrac{d\mathbf{x}}{dt}.  
  \end{align}
 Using \eqref{output:a} for $\aa=G(S,m)$, we obtain
 \begin{align*}
\dfrac{\partial a}{\partial m}
\;=\; \alpha N a(1-a),
\quad
\dfrac{\partial a}{\partial S_i}
= Na(a-1) r_i\;  \dfrac{1/K_I^i - 1/K_A^i}{(1+S_i/K_I^i)(1+S_i/K_A^i)} 
\;. 
 \end{align*}
For $i=1,2$, we assume that for any $\mathbf{x}$,
\begin{equation*}\label{assumption:KI:KA}
K_I^i\ll S_i (\mathbf{x})\ll K_A^i, 
\end{equation*}
as in \cite{tu2008modeling, jiang2010quantitative}.
This assumption guarantees scale-invariant behavior of \emph{E.~coli} in response to external signals, which was mathematically predicted in \cite{shoval2010fold} 
and experimentally verified in \cite{lazova2011response}.
Scale-invariance property of a system means that the system does not distinguish between an input (here, $S_1$ or $S_2$) and its scaled version (e.g., $p_1S_1$ or $p_2 S_2$). For more details, see
\cite{edgington2018mathematical} and \cite{shoval2011symmetry}. 
Using this assumption, we make the following approximation 
\[
 \dfrac{1/K_I^i - 1/K_A^i}{(1+S_i/K_I^i)(1+S_i/K_A^i)} \approx \dfrac{1}{S_i}. 
\]
Therefore, 
 \begin{align} \label{dynamics:a:part}
 \dfrac{da}{dt} \;=\;\alpha N a(1-a)\; \dfrac{a_0 - a}{\tau_a} +
Na(a-1)\Big( r_1\dfrac{\nabla_{\mathbf{x}}{S_1}\cdot\;{d\mathbf{x}}/{dt}}{S_1} +r_2\dfrac{\nabla_{\mathbf{x}}{S_2}\cdot\;{d\mathbf{x}}/{dt}}{S_2}  \Big).
 \end{align}
 
n the case of an one-dimensional space, we use the notation $S_i'=dS_i/dx$ ($i=1,2$) and denote $\gamma:=r_1/r_2$  denotes the ratio Tar/Tsr. Recall that $r_1+r_2=1$, so indeed $r_1=\frac{\gamma}{1+\gamma}$ and $r_2=\frac{1}{1+\gamma}$. 
Experimental data on the parameters used in this section are listed in Table \ref{Table:Parameters}.

 \begin{remark}
 \emph{E.~coli} bacteria can also sense pH changes, and their internal dynamics during pH taxis is analogous to that during chemotaxis. For example, according to \cite{hu2014behaviors,demir2011effects,yang2012opposite}, Tar receptors are attracted to a decrease of pH, but Tsr receptors show the opposite response. 
 Taking into account two chemical stimuli with different pH levels, we can apply the heterogeneous MWC model and use the following assumptions to derive the internal dynamics for pH:
 \begin{align*}
K_I^1\ll S_1 (\mathbf{x})\ll K_A^1,
\quad\mbox {and}  \quad
K_A^2\ll S_2 (\mathbf{x})\ll K_I^2, 
\end{align*}
which yield
\begin{align*}
 \dfrac{1/K_I^1 - 1/K_A^1}{(1+S_1/K_I^1)(1+S_1/K_A^1)} \approx \dfrac{1}{S_1}, 
\quad\mbox {and}  \quad
\dfrac{1/K_I^2 - 1/K_A^2}{(1+S_2/K_I^2)(1+S_2/K_A^2)} \approx \dfrac{-1}{S_2}. 
\end{align*}
 \end{remark}

\begin{remark}
In this work, we are interested in the total receptor kinase activity of the entire receptor cluster. Thus, 
we do not consider two different methylation dynamics for two different type of receptors as in \cite{hu2014behaviors}. 
\end{remark}

 As a result of the slow adaptation process \eqref{methylation:m}, 
bacteria use their methylation state as a short-term memory store to compare changes of stimuli temporarily during a run. 
This process helps the bacteria to run or tumble effectively toward their preferred location. {According to experimental observations and measurements,}
the tumbling rate function can be described as
\begin{align}\label{tumble_rate}
    \lambda (\aa ) = \lambda_{0} +  \frac{1}{\tau} \Big( \frac{a}{a_{0}} \Big)^{H},
\end{align}
where $\lambda_{0}, H,$ and $\tau$ denote the rotational diffusion, the Hill coefficient of flagellar motor's response curve, and the average run time, respectively, and $\aa_{0}$ is as given in \eqref{methylation:m}. 
Note that since $\aa$ depends on $S$, we may write $\lambda=\lambda(\aa,S)$ (see Section \ref{Fokker-Planck} below). 
More details about the physical meaning of these parameters can be found in
\cite{hu2014behaviors,jiang2010quantitative,sourjik2002receptor}.  The parameter values are shown in Table \ref{Table:Parameters}.

%****************
\subsection{Deriving a Fokker-Planck equation describing a population of bacteria}\label{Fokker-Planck}

In what follows, we describe the motion of a population of bacteria by incorporating their internal dynamics.

Let $p(\mathbf{x}, \bm{\aa},\bm{\nu}, t)$ be a probability density function describing a population of bacteria, modeled in a $2\mathcal{N}+\mathcal{M}+1$ dimensional phase space, where 
time $t\in\r$, $\mathbf{x}=(x_1,\ldots,x_{\mathcal{N}})\in \r^\mathcal{N}$ (we will specialize to
$\mathcal{N}=1, 2$) denotes the position of a cell centroid, 
$\bm{\aa}=(\aa_1, \ldots,\aa_{\mathcal{M}})\in A\subset \r^\mathcal{M}$ (we will specialize to
$\mathcal{M}=1$) denotes the internal dynamics of the
cell, and $\bm{\nu} = (\nu_1,\ldots, \nu_{\mathcal{N}})\in V\subset\r^\mathcal{N}$ denotes its velocity, $d\mathbf{x}/dt ={ \bm{\nu} }$.
The vector $\bm{S}(\mathbf{x}, t)=(S_1(\mathbf{x},t), \ldots, S_{\mathcal{K}}(\mathbf{x},t))\in \r^{\mathcal{K}}$ represents the concentration of extracellular signals in the environment (we will assume that $\bm{S}$ only depends on $\mathbf{x}$ as in Section \ref{internal:review} and \cite{kalinin2010responses}). 

Let the following system of ODEs
describe the evolution of the intracellular state, in the presence of the
extracellular signal $\bm{S}$:
\begin{equation}
\label{deterministic_model}
\frac{d \bm{\aa}}{dt}\;=\; f(\bm{\aa}, \bm{S}),
\end{equation}
where $f\colon \r^{\mathcal{M}}\times\r^{\mathcal{K}}\to \r^{\mathcal{M}}$ is a continuously differentiable
function with respect to each component, i.e.,  $f\in C^1(\r^{\mathcal{M}} \times \r^{\mathcal{K}})$. 

Assuming constant velocity, $d\nu_i/dt=0$, the evolution of $p=p(\mathbf{x}, \bm{\aa}, \bm{\nu}, t)$ with turning rate $\lambda=\lambda(\bm{\aa}, \bm{S})$ is governed by
the following forward Fokker-Planck 
equation describing a velocity-jump process \cite{alt1980biased, Othmer1988ModelsOD}: 
\begin{equation}
\label{transport}
\frac{\partial p}{\partial t}+\nabla_{\mathbf{x}}\cdot \bm{\nu} p +\nabla_{\bm{\aa}}\cdot f p = -\lambda(\bm{\aa}, \bm{S})p +\displaystyle\int_V \lambda(\bm{\aa}, \bm{S}) T(\bm{\aa}, \bm{\nu}, \bm{\nu}')p(\mathbf{x}, \bm{\aa}, \bm{\nu}', t)\;d \bm{\nu}',
\end{equation}
where the non-negative kernel $T(\bm{\aa}, \bm{\nu}, \bm{\nu}')$ is the probability that the bacteria changes the velocity from $\bm{\nu}'$ to $\bm{\nu}$, and 
\[\displaystyle\int_V T(\bm{\aa}, \bm{\nu}, \bm{\nu}') \;d \bm{\nu}'=1.\]

Equation \eqref{transport} is not tractable mathematically and is hard to be validated by typical experimental techniques. 
The goal is to use the \textit{microscopic} model (\ref{transport}), and  derive a \textit{macroscopic} model  for chemotaxis in a one-dimensional space (in Section \ref{AD:1dim}) and a two-dimensional space (in Section \ref{AD:2dim}), i.e., an equation for the marginal density
\begin{equation*}
n(\mathbf{x}, t)\;=\;\displaystyle\int_V\displaystyle\int_A p(\mathbf{x}, \aa, \bm{\nu}, t)\;d\aa\; d\bm{\nu}, 
\end{equation*}
with 
$\mathcal{N}=1$ or 2, $\mathcal{M}=1$, and $\mathcal{K}=2$;
$n({\bf x},t)$ is the number  of individuals which at time $t$ are located at position $\bf{x}$, whatever their internal dynamics and  velocity are.

Note that our theory works for any arbitrary $\mathcal{K}$. However, we are interested in 
two extracellular signals, so we only consider 
$\mathcal{K}=2$.

%%% Section 3 
\section{Advection-diffusion equation for chemotaxis in response to two stimuli in a one dimensional space}\label{AD:1dim}

In this section, we assume that the bacteria move in a one-dimensional space, i.e., a finite interval $[0,L]$ where we assume $L$ is sufficiently large. We let $p^{\pm}(x, \aa, t)= p(x, \aa,\pm\nu,  t)$ denote the density of the bacteria, located at $x\in[0,L]$,  moving to the right and left, respectively;  and let 
$f^{\pm} = f_0\pm \nu f_1$ describe their corresponding internal state. Here, $\nu>0$ represents the speed of the bacteria, and we assume that $\nu$ is constant. Then the Fokker-Planck equation (\ref{transport}) becomes 
\begin{align}
\frac{\partial p^+}{\partial t}\;+\;&\nu\frac{\partial p^+}{\partial x} \;+\; \frac{\partial }{\partial \aa} \left[f^{+}(\aa, \bm{S})\; p^+\right]\label{transport_one_dim_general:1}
\;=\; \dfrac{1}{2}\;\lambda(\aa, \bm{S}) (p^--p^+),\\
\frac{\partial p^-}{\partial t}\;-\;&\nu\frac{\partial p^-}{\partial x} \;+\; \frac{\partial }{\partial \aa} \left[f^{-}(\aa, \bm{S})\; p^-\right]\label{transport_one_dim_general:2}
\;=\;\dfrac{1}{2}\; \lambda(\aa, \bm{S}) (p^+-p^-). 
\end{align}
Following  \cite{aminzare2013remarks}, under a decay condition for $p^{\pm}$, some conditions on the internal dynamics (for example, shallow conditions for the stimuli-- see Proposition \ref{Prop:AD:1D} below), moment closure techniques, and parabolic scaling, a general advection-diffusion equation for the marginal density
\begin{equation*}
n(x, t)\;=\;\displaystyle\displaystyle\int_A (p^+(x, \aa, t)+p^-(x, \aa, t))\;d\aa
\end{equation*}
can be derived from Equations \eqref{transport_one_dim_general:1}-\eqref{transport_one_dim_general:2} as follows
\begin{equation}\label{advection_diffusion_1D}
\frac{\partial n}{\partial t}\;=\; \frac{\partial}{\partial x}\left(\frac{\nu^2}{\alpha_0}\frac{\partial n}{\partial x}-\frac{\alpha_1B_0\nu^2}{\alpha_0(A_1-\alpha_0)}\; n\right).
\end{equation}

Here, $\alpha_i$,  $A_i$, and $B_i$ are the Taylor constants of $\lambda$, $f_0$,  and $f_1$, respectively: 
\begin{align*}
\lambda \;=\; \alpha_0 +\alpha_1 \aa + \cdots,\\ 
f_0\;=\;A_0+A_1 \aa+\cdots,
\\
f_1\;=\;B_0+B_1 \aa+\cdots. 
\end{align*}
All the Taylor constants depend on $\bm{S}= \bm{S}(x)$ and we assume that $A_0=0$, $A_1\neq0$, $a_0\neq0$, $A_1\neq a_0$, and $B_0\neq0$. 
We omit the derivation of the one-dimensional advection-diffusion equation~\eqref{advection_diffusion_1D}, since the derivation is very similar to (and easier than) the two-dimensional advection-diffusion equation~\eqref{advection_diffusion_2D}, which is given in Section~\ref{AD:2dim} below. 

\begin{remark}
In \cite{aminzare2013remarks}, the authors assumed that the non-negative kernel $T(\aa, \nu, \nu')$ is the probability that the bacteria changes the velocity from $\nu'$ to $\nu$, if a change of direction occurs.  Therefore, in a one-dimensional space, $T(\aa, \nu, \nu')=1$, and hence the right hand side of  \eqref{transport_one_dim_general:1}-\eqref{transport_one_dim_general:2} 
 for \cite{aminzare2013remarks}
has no factor $1/2$. In this work, we do not assume such an assumption; therefore  $T(\aa, \nu, \nu')=1/2$. The assumption in \cite{aminzare2013remarks} leads to the following equation instead of  \eqref{advection_diffusion_1D}:
\begin{equation*}\label{advection_diffusion_1D_old}
\frac{\partial n}{\partial t}\;=\; \frac{\partial}{\partial x}\left(\frac{\nu^2}{2\alpha_0}\frac{\partial n}{\partial x}-\frac{\alpha_1B_0\nu^2}{\alpha_0(A_1-2\alpha_0)}\; n\right). 
\end{equation*}

In \cite{erban2004individual, xue2009multiscale}, Equation \eqref{advection_diffusion_1D} is derived for a toy model
that captures the essential excitation and adaptation components.
Here,  \eqref{advection_diffusion_1D} can be used for any continuous tumbling function $\lambda$ and a larger class of internal dynamics $f^\pm= f_0\pm \nu f_1$, (see the following section for more details).
\end{remark}

%****************
\subsection{Application to a population of \emph{E.\ coli} bacteria}
\label{AD-application:1dim}

In what follows, we determine the terms in the advection-diffusion equation \eqref{advection_diffusion_1D} for a population of \emph{E.\ coli} bacteria in a spatial domain $[0,L]$ equipped with two chemical gradients MeAsp, denoted by $S_1(x)$,  and serine, denoted by $S_2(x)$. We further assume that $S_1$ and $S_2$ are respectively  increasing and decreasing functions on $[0,L]$, i.e., MeAsp accumulates near $x=L$ and serine accumulates near $x=0$. 
As we discussed in Section \ref{internal:review}, in a one-dimensional space,  the internal state of \emph{E.\ coli} evolves according to the
following {ODE}:
\begin{equation}
\label{Example:internal:a:1D}
\frac{d\aa}{dt}\;=\; f^\pm(\aa, S_1,S_2) = f_0(\aa,S_1,S_2)\pm \nu f_1(\aa,S_1,S_2),
\end{equation}
where, as described in \eqref{dynamics:a:part},  
\begin{align}\label{Example:internal:a:coeff:1D}
\begin{aligned}
f_0 (\aa,S_1,S_2)&= \dfrac{\alpha}{\tau_a}N \aa(\aa-\aa_0)(\aa-1), \\
f_1(\aa,S_1,S_2)&=N \aa(\aa-1)\left(\dfrac{\gamma}{1+\gamma} \dfrac{S_1'}{S_1}+\dfrac{1}{1+\gamma} \dfrac{S_2'}{S_2}\right).
\end{aligned}
\end{align}
Here, $S_i'=dS_i/dx$ and  $\gamma:=r_1/r_2$ denotes the ratio Tar/Tsr. Recall that $r_1+r_2=1$, so indeed $r_1=\frac{\gamma}{1+\gamma}$ and $r_2=\frac{1}{1+\gamma}$. 
All the parameters used in this section are described in Section \ref{internal:review}.

\begin{proposition}\label{Prop:AD:1D}
 Assume that the density functions $p^\pm$ satisfy the decay condition
\begin{align*}
    p^{\pm} (x,\aa, t) \leq C(x,t) e^{-c (x,t)  \aa} 
\end{align*}
for some functions $C, c : \mathbb{R} \times [0,\infty) \rightarrow \mathbb{R}_{>0}$
and the stimuli $S_1$ and $S_2$ satisfy the shallow condition
\begin{equation}\label{Example:shallow:condition:1D}
\left|\dfrac{\gamma}{1+\gamma} \dfrac{S_1'(x)}{S_1(x)}+\dfrac{1}{1+\gamma} \dfrac{S_2'(x)}{S_2(x)}\right|\;\leq\; \min\{q, 1-q\}\;\dfrac{p}{\nu}, \qquad \forall x\in[0,L], 
\end{equation}
 where $q=a_0$ and $p=\frac{\alpha}{\tau_a}$ represent the adapted value and the the speed of adaptation, respectively. 
Then, for the given internal dynamics \eqref{Example:internal:a:1D}, the dynamics of a population of \emph{E. coli}, $n(x,t)$, can be approximated by the advection-diffusion
\begin{equation}\label{main_equation_example}
\frac{\partial n}{\partial t}\;=\;\frac{\partial }{\partial x}\left(D\;\frac{\partial n}{\partial x}-\chi\;\left(\dfrac{\gamma}{1+\gamma} \dfrac{S_1'}{S_1}+\dfrac{1}{1+\gamma} \dfrac{S_2'}{S_2}\right)n\right),
\end{equation}
where the diffusion coefficient $D$ and the advection constant $\chi$ are as follows:
\begin{align}\label{main_equation_example_coeff}
    D = \frac{\nu^{2}}{\lambda_{0} + r q^{H}},
 \qquad \chi = \frac{r NH q^{H}(q-1)\nu^{2}}{(\lambda_{0} + rq^{H})(Npq(q-1) - \lambda_{0} - r q^{H})}\;.
\end{align}
\end{proposition}
\begin{proof}
The proof is similar to the case of one stimulus, see \cite{aminzare2013remarks, xue2015macroscopic},
and the case of two-dimensional space which is given in Section \ref{AD:2dim} below. 
\end{proof}

Note that the condition \eqref{Example:shallow:condition:1D} holds if either the adaptation rate $p$ is large or $\gamma$, $S_1$ and $S_2$ are chosen so that the left hand side (LHS) of \eqref{Example:shallow:condition:1D} is small, i.e., the shallow condition is equivalent to either small changes in the environment or fast adaptation. See the examples given in Section \ref{simulation:1dim} for more details. 

Now we determine the boundary conditions of  \eqref{main_equation_example}.  Following the experimental set up in \cite{kalinin2010responses}, we want the population of the bacteria to be conserved in time, i.e., for any $t \geq 0$, 
\begin{align}\label{n-constant}
   0=  \dfrac{d}{dt} \int_{0}^{L} n(x,t) dx = D\Big( \dfrac{\partial n}{\partial x}(L,t) - \dfrac{\partial n}{\partial x}(0,t) \Big) - \chi \big( V(L) n(L,t) - V(0) n(0,t) \big),
\end{align}
where 
\begin{align}\label{Robin_condition_V}
    V(x) =   \dfrac{\gamma}{1+\gamma} \dfrac{S_{1}'(x)}{S_{1}(x)} + \dfrac{1}{1+\gamma} \dfrac{S_{2}'(x)}{S_{2}(x)}.
\end{align} 
 The following zero flux boundary conditions at $x =0$ and $x= L$ guarantee \eqref{n-constant}. For any $t \geq 0$, 
\begin{align}\label{Robin_condition}
 \dfrac{\partial n}{\partial x}(0,t) =\dfrac{\chi}{D}  V(0) n(0,t) \quad \mbox{and} \quad 
 \dfrac{\partial n}{\partial x} (L,t) = \dfrac{\chi}{D}  V(L) n(L,t). 
\end{align}

In the following lemma, we provide sufficient conditions which guarantee existence and uniqueness of solutions of  \eqref{main_equation_example} with boundary conditions \eqref{Robin_condition}. 

\begin{lemma}\label{lem:existence 1D}
Let $V(x)$ be continuous on $[0,L]$ and  $n_{0}(x)$ be a smooth non-negative function. Then, \eqref{main_equation_example} with boundary condition \eqref{Robin_condition} and initial condition $n_{0}(x)$ admits a unique solution of the form  
$n(x,t) = \sum_{n=1}^{\infty} X_{n}(x) T_{n}(t)$. Moreover, $n(x,t)$ is uniformly bounded in $x$ and $t.$
\end{lemma}

This lemma can be proved by the method of separation of variables in a standard way: We can apply Sturm-Liouville theory \cite{kapitula2013spectral} to solve the eigenproblem in which the first eigenvalue can be also explicitly estimated to guarantee the uniform boundedness of the solution in time. 
For a proof see Appendix \ref{Appendix0}.

%****************
\subsection{Steady state solution of advection-diffusion equation with zero flux boundary conditions}\label{steady-state-1D}

The bacterial responses to MeAsp and serine depend on the ratio of their chemoreceptors Tar and Tsr, i.e.,  $\gamma=$ Tar/Tsr. 
The goal is to find a positive $\gamma^*$ and show that for $\gamma>\gamma^*$ the bacteria tend to move toward a gradient of increasing MeAsp  (i.e., accumulate near $x=L$) and  for $\gamma<\gamma^*$ they move toward a gradient of increasing serine (i.e., accumulate near $x=0$). 
To determined such a $\gamma^{*}$, we look at a steady state of advection-diffusion equation \eqref{main_equation_example} with boundary condition \eqref{Robin_condition}.

Let $\Phi(x)$ be the steady state solution of the advection-diffusion equation \eqref{main_equation_example} with boundary condition \eqref{Robin_condition}.
If $S_{1},S_{2},$ and $\gamma$ are chosen such that $V(x)$ satisfies the condition in Lemma \ref{lem:existence 1D}, 
then the solution of \eqref{main_equation_example} converges to  $\Phi(x)$ as $t \rightarrow \infty.$ Indeed, in the following examples, $V(x)$ satisfies the condition in Lemma \ref{lem:existence 1D}.

Assuming that the bacteria start from a point $x_0\in(0,L)$, they move toward a gradient of increasing MeAsp (respectively, serine) and accumulate near $x=L$ (respectively, $x=0$), if the steady state solution of the advection-diffusion equation \eqref{main_equation_example} admits a maximum on the right (respectively, left) sub-interval $(x_0,L]$ (respectively, $[0, x_0)$). 
Therefore, in what follows, we find conditions that $\Phi(x)$ admits a maximum on the right  sub-interval $(x_0,L]$ or the left sub-interval $[0, x_0)$. 
 
To compute the steady state solution of  \eqref{main_equation_example}, we let ${\partial n}/{\partial t} = 0$, which gives a constant flux, i.e.,
$D\;{\partial n}/{\partial x}-\chi V(x) n=\text{constant}.$
Assuming zero flux boundary conditions \eqref{Robin_condition}, the constant becomes zero and a simple calculation shows that the steady state solution satisfies
\begin{align}\label{steady_state}
\Phi(x) = \Phi(c_{0}) \exp\left\{\frac{\chi}{D}\displaystyle\int_{c_{0}}^{x} V(y) dy\right\}.
\end{align}
We choose $c_0$ such that $\Phi(c_0)>0$. 
Indeed, there is such a $c_{0}$  by \eqref{n-constant}:
\begin{align*}
\frac{d}{dt} \int_0^L n(x,t)dx =0\quad&\Rightarrow\quad
\int_0^L n(x,t)dx = \text{constant}>0\\
&\Rightarrow\quad\lim_{t\to\infty}\int_0^L n(x,t)dx = \int_0^L \Phi(x) dx = \text{constant}>0\\
&\Rightarrow\quad\text{there exists $c_0$ such that $\Phi(c_0)>0$. }
\end{align*}

 In what follows, we write $V$ as a function of both $x$ and $\gamma$, $V=V(x,\gamma)$. 
Considering the fact that  $\Phi'(x)= \frac{\chi}{D}V(x,\gamma)\Phi(x)$ and $\Phi(x)>0$,  $\Phi$ takes a unique maximum at $x^*\in[0,L]$ if, for any $\gamma>0$, either $V$ does not change sign or $V$ is a non-increasing function of $x$ and $V(x^*,\gamma)=0$. 
Now we are ready to find $\gamma^*$ in the following lemma. 

\begin{lemma}\label{bifurcation_value}
Assume the bacteria start at $x_0\in[0,L]$ and for any $\gamma>0$, $\partial V/\partial x\leq0$. Also, assume that $S_1$ and $S_2$ are respectively increasing and decreasing functions on $[0,L]$. 
Then there exists $\gamma^*>0$ such that $V(x_0,\gamma^*) =0$ and for $\gamma>\gamma^*$ the bacteria accumulate on the right side of $x_0$ and for $\gamma<\gamma^*$ they accumulate on the left side. 
\end{lemma}

\begin{proof}
A simple calculation shows that
$V(x,\gamma) = 0$ if and only if \[\gamma(x) = \dfrac{S'_2(x)/S_2(x)}{S'_1(x)/S_1(x)}.\] 
Let $\gamma^*:=\gamma(x_0).$ 
For $\gamma>\gamma^*$, $V(x_0, \gamma)>0$, therefore, since  $\partial V/\partial x\leq0$, $\Phi$ takes its maximum (either $x=L$ or $x=x^*<L$) on the right side of $x_0$, and hence the bacteria accumulates toward the right side of $x_0$. Similarly, if $\gamma<\gamma^*$, $V(x_0, \gamma)<0$, and hence $\Phi$ takes its maximum (either $x=0$ or $x=x^*>0$) on the left side of $x_0$, and hence the bacteria accumulates toward the left side of $x_0$.
\end{proof}

We refer to $\gamma$ and $\gamma^*$ as the bifurcation parameter and bifurcation value, respectively, since at $\gamma=\gamma^*$ the direction of the bacterial changes. See Figure \ref{fig:V} below.  

Note that if the bacteria are initially distributed on  $[0,L]$ instead of locating on a single point $x_0$, we consider $\gamma^*=\gamma^*(L/2)$ as the bifurcation value. 

In the following section, we consider two  sets of stimuli: (i) $S_1$ linear and increasing, $S_2$ linear and decreasing; (ii)  $S_1$ exponential and increasing, $S_2$ exponential and decreasing. We also assume that the bacteria are located at $x_0=L/2$ initially.  In both cases, $V(x,\gamma)$ is a decreasing function on $[0,L]$. Hence, the conditions of Lemma  \ref{bifurcation_value} hold and, therefore, $\gamma^*$ can be determined based on the initial location of the bacteria, i.e.,  $x_0=L/2$.

%%% Section 4
\section{Monte-Carlo agent-based simulations in a one-dimensional space}\label{simulation:1dim}

To show that the advection-diffusion equation \eqref{main_equation_example} with boundary condition \eqref{Robin_condition} is a good approximation for the microscopic description of \emph{E. coli} chemotaxis, 
 we run a Monte-Carlo agent-based simulation. {A detailed description of the Monte-Carlo simulation is given in  Appendix \ref{Appendix2}}. 

The following computational setting of our Monte-Carlo agent-based simulation is motivated by the experimental set up in \cite{kalinin2010responses}.

\begin{description}[leftmargin=*]
\item[Spatial Domain.] A one-dimensional channel of length of $400 \mu m$ ($x \in [0,400]$). 

\item[Stimuli.]
Along the two sides of the channel 
two opposing chemical signals, $S_{1}(x)$ and $S_{2}(x)$, 
flow and diffuse across the channel. Two opposing linear and two opposing exponential chemical signals are considered in Sections \ref{lin_ligands:1D} and \ref{exp_ligands:1D}, respectively.

\item[Initial Condition.] At $t=0$ (sec), an ensemble of 100,000 agents is located in the center of the channel ($x=200$).

\item[Boundary Conditions.] When a cell reaches a boundary, we relocate the cell to stay inside the domain, i.e., zero flux boundary condition is applied.

\item[Simulation Duration.]
We simulate the bacterial behavior for 200 sec, $t \in [0,200]$. 
It is observed that the solution of each simulation in this section becomes stationary at $t=200.$
\end{description}

To illustrate distributions of the cells, we display histograms with 100 equal-sized bins. 

We use an explicit finite difference method to numerically solve the advection-diffusion equation \eqref{main_equation_example} with the  boundary condition \eqref{Robin_condition}.  

In the following examples, we compare the solutions of the macroscopic equation \eqref{main_equation_example} with boundary conditions \eqref{Robin_condition} with results of the Monte-Carlo simulation. Further, for each case, we compute the bifurcation value $\gamma^*$ defined in Section \ref{steady-state-1D}.
To measure bacterial  preference, we define the chemotactic migration coefficient (CMC):
\begin{align}\label{CMC}
    \text{CMC}_x(t) = \frac{\text{mean}(x(t)) - 200}{200}.
\end{align}
In the Monte-Carlo simulation, mean($x(t)$) is  the average of individual positions $x_{i}$ at time $t$ across the channel, i.e.,  $\text{mean}_i(x_{i}(t))$. For a solution $n(x,t)$ of \eqref{main_equation_example}, mean($x(t)$) is the expectation value of the probability density $n(x,t)$, i.e., $\int_0^L xn(x,t) dx$. 
The absolute value of $\text{CMC}_x$ determines the displacement of the   bacteria in $x$-direction. 
The sign of $\text{CMC}_x$
 indicates their preference to the right or left. 
When $\text{CMC}_x>0$ (respectively, $\text{CMC}_x<0$), the bacteria tend to move to the right, i.e., above $x=200$ (respectively, left, i.e., below $x=200$). 

%****************
\subsection{Chemotaxis in response to two linear gradients}\label{lin_ligands:1D}

To demonstrate responses of \emph{E. coli} to two opposing linear gradients MeAsp and serine, and following the experimental set up in \cite{kalinin2010responses},  we let
\begin{equation}\label{linear_linear_gradients}
S_{1}(x)=0.5 x + 130 \quad \mbox{and}\quad  S_{2}(x)=-0.03x + 20
\end{equation}
 represent the concentrations of
MeAsp and serine at each point $x\in[0,400]$, respectively.

As we discussed in Section \ref{steady-state-1D}, since for any $\gamma>0$,
\begin{align}\label{V_1d_lin_lin}
V(x,\gamma)=\dfrac{\gamma}{1+\gamma} \dfrac{0.5}{0.5x+130}+\dfrac{1}{1+\gamma} \dfrac{-0.03}{-0.03x+20}
\end{align}
is decreasing  on $[0,400]$,  $V(200,\gamma)$ is an increasing function of $\gamma$, and $V(200,\gamma^*)=0$ for $\gamma^*\approx 0.985$, by Lemma \ref{bifurcation_value}, for $\gamma>0.985$ (respectively, $\gamma<0.985$)
the bacteria move to the right (respectively, left), toward the gradient of
MeAsp (respectively, serine). 

\begin{remark}
In this example, for any $x \in [0,L]$, $S_{1} > S_{2}$, $|S_{1}'| > |S_{2}'|$ and $|S_{1}'/S_{1}| > |S_{2}'/S_{2}|$. Therefore, one may expect that the bacterial always choose to move towards MeAsp ($S_{1}$). However, as we proved in Lemma \ref{bifurcation_value},  when  the ratio Tar/Tsr is small enough ($\gamma < \gamma^{*}$),  the bacteria move toward the gradient of serine. 
Figure \ref{fig:V} displays the relation between $\gamma$ and the initial position of the bacteria, $x_{0}.$ The dotted curve $\gamma^{*}(x_{0}) = \big(\frac{S_{2}'/S_{2} } { S_{1}'/S_{1} }\big)(x_{0}) = \frac{780 + 3 x_{0}}{2000-3x_{0}}$ satisfying $V(x_{0}, \gamma^{*})=0$ represents the bifurcation values in which the bacterial direction changes. As it is shown in Figure \ref{fig:V}, $\gamma^{*}$ is an increasing function in $x_{0}$, that is, $|S_{1}'/S_{1}|$ increases faster than $|S_{2}'/S_{2}|$ as $x_{0}$ increases. This means that if the bacteria start from near the right end point, a stronger force (a larger $\gamma^{*}$) is needed to drag them toward the gradient of serine ($S_{2}$). In the following section with exponential gradients, although $S_{1}>S_{2}$ and $S_{1}' > S_{2}'$ everywhere, the needed force $\gamma^{*}$ to drag the bacteria to the gradient of serine is always equal to 1. The  reason is that $|S_{1}'/S_{1}| \equiv |S_{2}'/S_{2}|$, in that case.

%******Fig 2******
\begin{figure}[ht]
\floatbox[{\capbeside\thisfloatsetup{capbesideposition={right,top},capbesidewidth=6cm}}]{figure}[\FBwidth]
{\includegraphics[width=7.55cm]{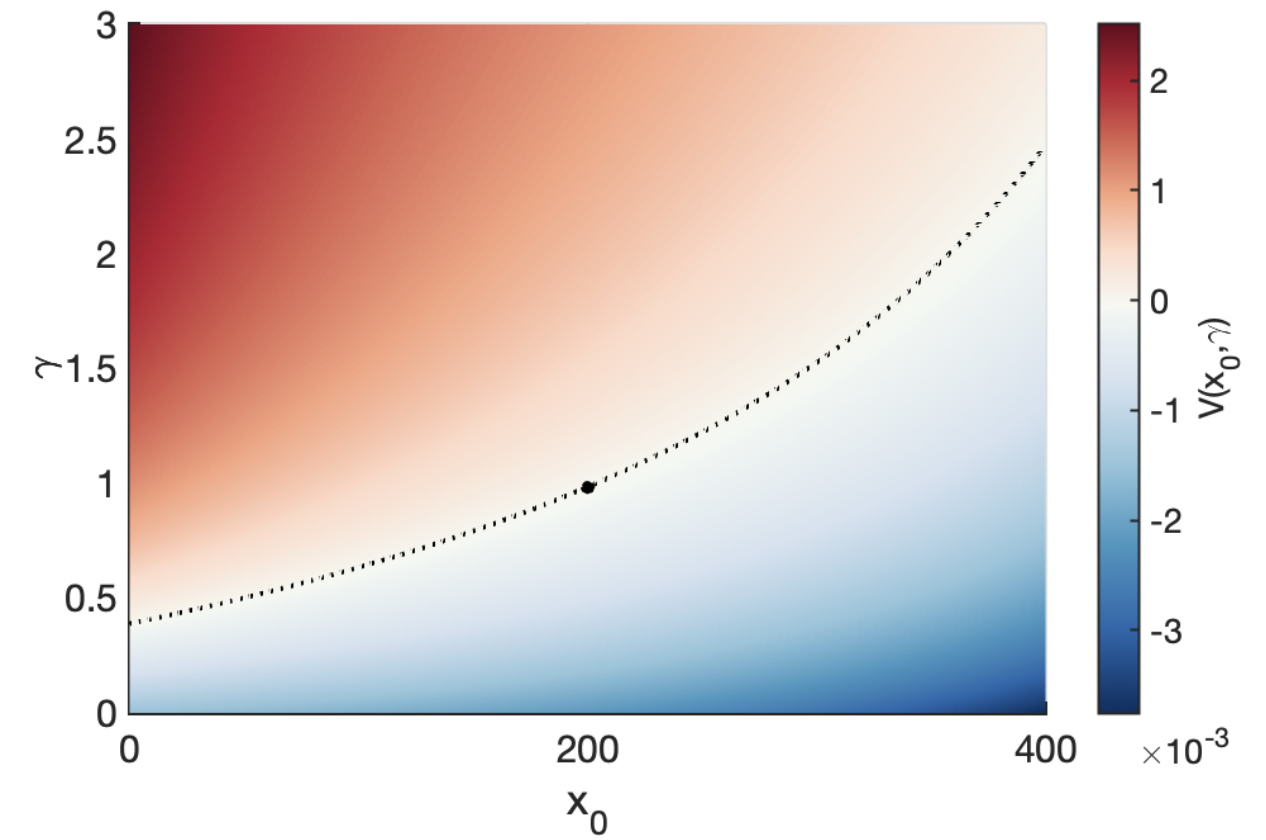}}
{\hspace{0.7cm} \caption{\small{Change of signs of $V$ in (\ref{V_1d_lin_lin}) as $x_0$ and $\gamma$ vary.
For $(x_{0},\gamma)$ in the dark red (respectively, blue) region, 
$V$ becomes positive (respectively, negative) as shown in the color bar. The dotted curve is a set of $(x_{0},\gamma^{*})$ where $V=0$.
The solid point at $(200,0.985)$ indicates the bifurcation value for the simulation in Section \ref{lin_ligands:1D}.}
}\label{fig:V}}
\end{figure}
\end{remark}

To examine the result of Lemma \ref{bifurcation_value}, we choose two values for $\gamma$, $\gamma=1.5>\gamma^*\approx 0.985$ and $\gamma=0.5<\gamma^*\approx 0.985$, and in 
Figures~\ref{1d_lin_1p5_0p5}(a, c) display distributions of the normalized density of \emph{E. coli} obtained from the Monte-Carlo agent-based simulation and numerical solution of the advection-diffusion \eqref{main_equation_example}. 
Three snapshots at times $t=10, 60,200$ (sec) are shown. 
As expected, the snapshots of  a solution of  \eqref{main_equation_example} 
and 
 the snapshots of a solution of Monte-Carlo simulation move to the right  
when $\gamma>\gamma^*$, as shown in Figure~\ref{1d_lin_1p5_0p5}(c),  
and they move to the left 
when $\gamma<\gamma^*$, as shown in  Figure~\ref{1d_lin_1p5_0p5}(a). 

Figures~\ref{1d_lin_1p5_0p5}(b, d) display 
the corresponding $\text{CMC}_x$ which, as expected, is positive when $\gamma>\gamma^*$ and the bacteria accumulates on the right and is negative  when $\gamma<\gamma^*$ and the bacteria accumulates on the left. 

In Figure~\ref{1d_lin_1p5_0p5}, the adaptation speed rate $p$ is $0.4$ and other parameters are as given in Table \ref{Table:Parameters} (see Appendix \ref{Appendix1}). 
For the given linear stimuli, the  values of $\gamma$ and $p$ are chosen such that the shallow condition \eqref{Example:shallow:condition:1D} holds. 
Therefore, by Proposition \ref{Prop:AD:1D},  the advection-diffusion equation  \eqref{main_equation_example}
approximates the Fokker-Planck equations \eqref{transport_one_dim_general:1}-\eqref{transport_one_dim_general:2}. A comparison between numerical solutions of  \eqref{main_equation_example} and 
the solutions of Monte-Carlo simulations in Figure~\ref{1d_lin_1p5_0p5} confirms this result.
%******Fig 3******
\begin{figure}[ht]
	\begin{minipage}[c][1\width]{
	   0.45\textwidth}
	   \centering
	   \includegraphics[width=1\textwidth]{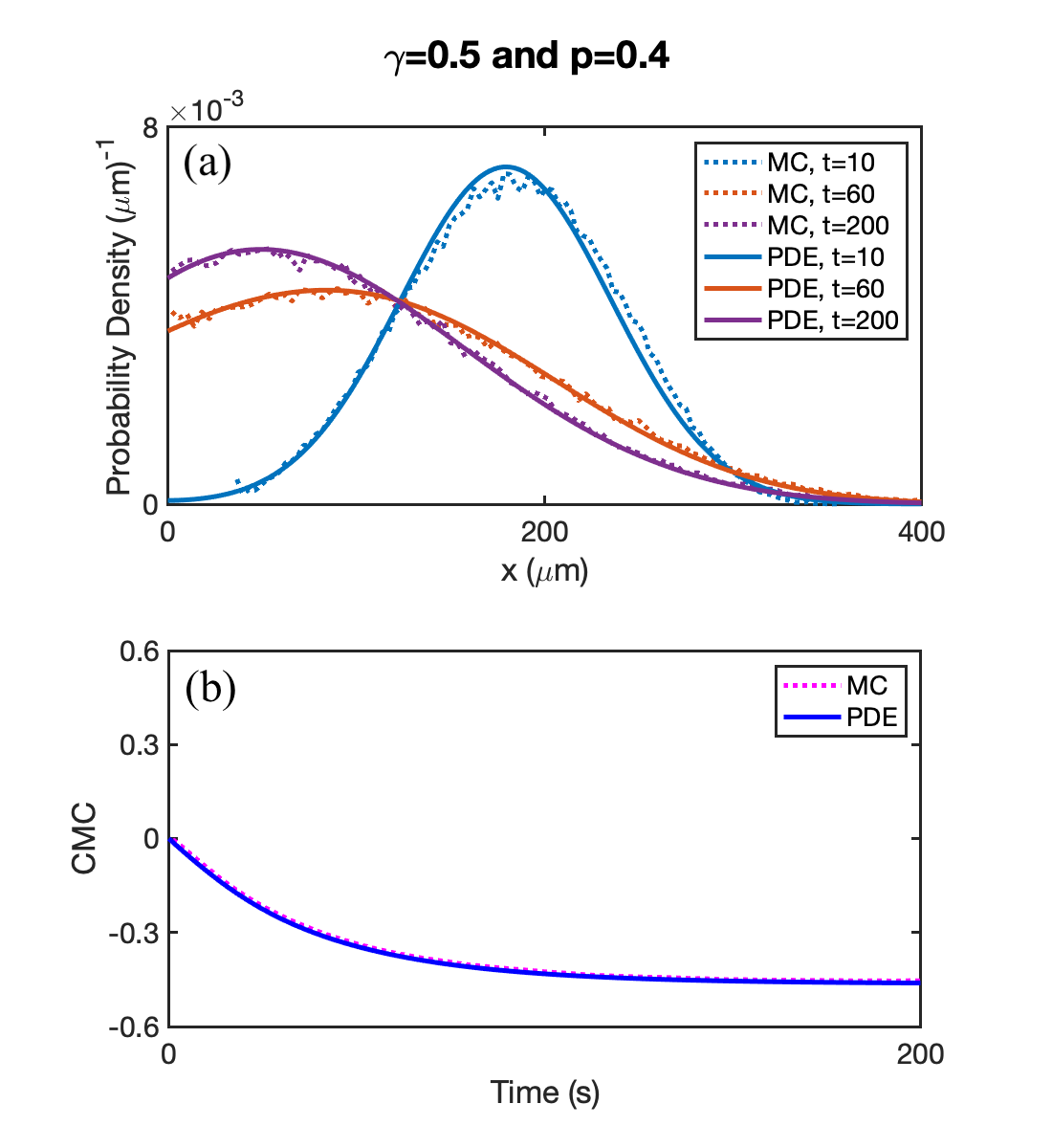} 
	\end{minipage}
	\begin{minipage}[c][1\width]{
	   0.45\textwidth}
	   \centering
	   \includegraphics[width=1\textwidth]{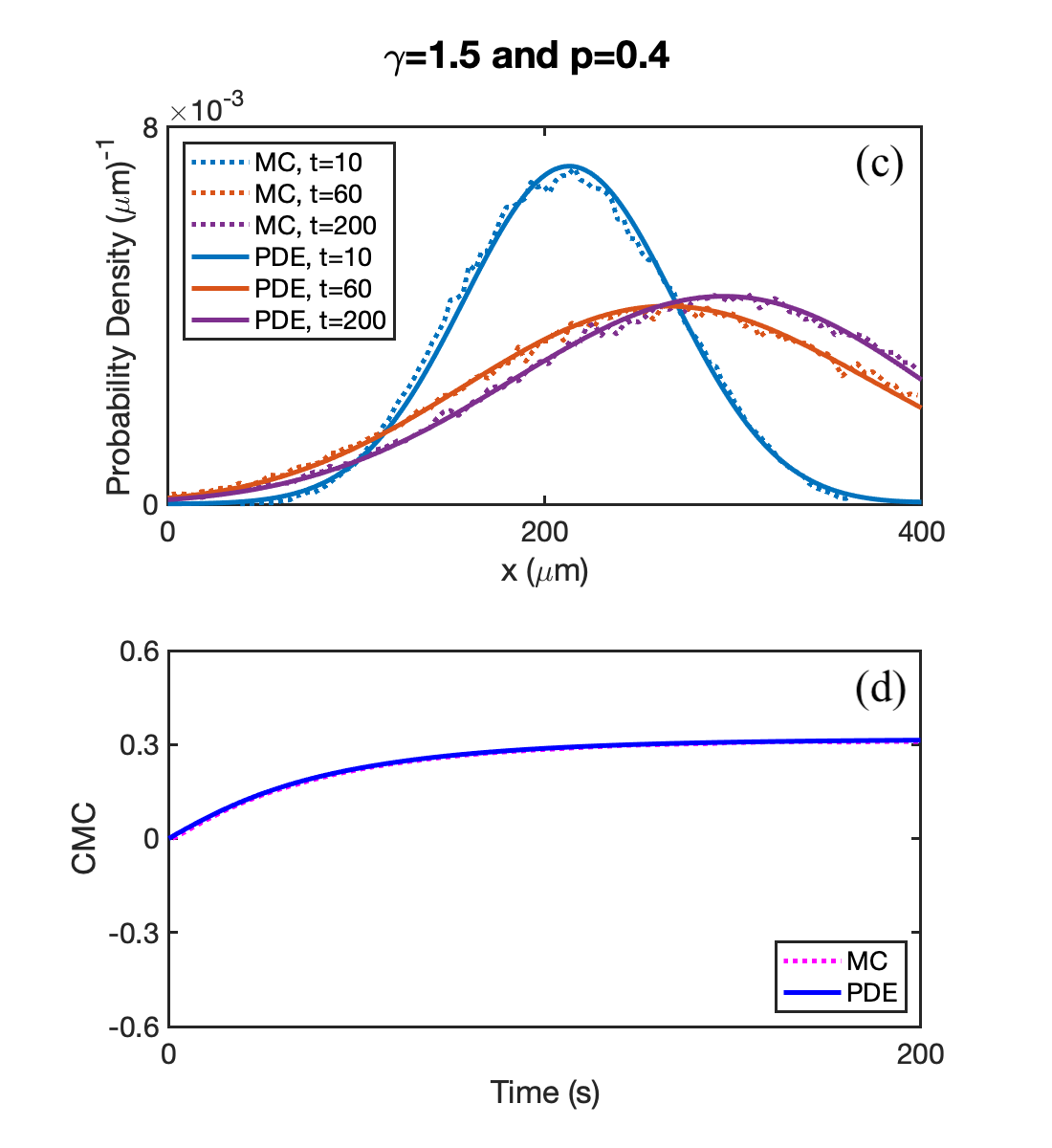}
	\end{minipage}
    \caption{
    (a) and (c): Comparisons of Monte-Carlo simulation and numerical solutions of  \eqref{main_equation_example} for two linear gradients \eqref{linear_linear_gradients} at times $t=10, 60,200$ (sec) with  $\gamma=0.5<\gamma^*$ and  $\gamma=1.5>\gamma^*$, in which the snapshots move to the left and right, respectively. (b) and (d): Comparisons of the corresponding $\text{CMC}_x$. }
    \label{1d_lin_1p5_0p5}
\end{figure}

%****************
\subsection{Chemotaxis  in response to two exponential gradients}\label{exp_ligands:1D}

We now repeat the discussion of Section \ref{lin_ligands:1D} for two opposing exponential gradients MeAsp and serine. We let
\begin{align}\label{exp_exp_gradients}
    S_{1}(x) = 130 e^{0.0023 x} \quad \mbox{and}\quad S_{2}(x) = 8 e^{-0.0023(x-400)}
\end{align}
 represent the concentrations of
MeAsp and serine at $x\in[0,400]$, respectively. Exponential gradients have been used for various chemotaxis environments (e.g., \cite{kalinin2009logarithmic}).

To find the bifurcation value $\gamma^*$, which determines the direction of bacteria, we apply Lemma \ref{bifurcation_value}.  
A simple calculation shows that $V$ of this example is equal to
\[V(x,\gamma)=0.0023\; \frac{\gamma-1}{\gamma+1}\;.\]
For any $\gamma>0$, $V(x,\gamma)$ is non-increasing  on $[0,400]$. Further,  $V(200,\gamma)$ is an increasing function of $\gamma$ and $V(200,\gamma^*)=0$ for $\gamma^*=1$. Therefore,  by Lemma \ref{bifurcation_value}, for $\gamma>1$
 the bacteria move to the right, toward the gradient of MeAsp, and for  $\gamma<1$
they move to the left, toward the gradient of serine. 

To examine this result, we choose two values for $\gamma$,  $\gamma=1.1>\gamma^*=1$ and $\gamma=0.9<\gamma^*=1$. 
Figures~\ref{1d_exp_1p1_0p9}(a, c) 
display distributions of the normalized density of \emph{E. coli} obtained from the Monte-Carlo agent-based simulation and numerical solution of the advection-diffusion \eqref{main_equation_example}.
Three snapshots at times $t=10, 60,200$ (sec) are shown. 
As expected, the snapshots of  a solution of  \eqref{main_equation_example} 
and the snapshots of a solution of Monte-Carlo simulation move to the right 
when $\gamma>\gamma^*$, as shown in Figures~\ref{1d_exp_1p1_0p9}(c)
  and they move to the left when $\gamma<\gamma^*$, as shown in 
  Figures~\ref{1d_exp_1p1_0p9}(a).
  
  Figures~\ref{1d_exp_1p1_0p9}(b, d) display 
  the corresponding $\text{CMC}_x$ which, as expected, is positive when $\gamma>\gamma^*$ and the bacteria accumulates on the right and is negative  when $\gamma<\gamma^*$ and the bacteria accumulates on the left. 
  
  Note that for the given exponential stimuli, the  values of $\gamma$ and $p=0.05$ are chosen such that the shallow condition \eqref{Example:shallow:condition:1D} holds. 
  As discussed in Section \ref{lin_ligands:1D},  Proposition~\ref{Prop:AD:1D} and Figure~\ref{1d_exp_1p1_0p9} confirm that the numerical solutions of \eqref{main_equation_example}
agree well with  
the solutions of Monte-Carlo simulations.
%******Fig 4******
\begin{figure}[ht]
	\begin{minipage}[c][1\width]{
	   0.45\textwidth}
	   \centering
	   \includegraphics[width=1\textwidth]{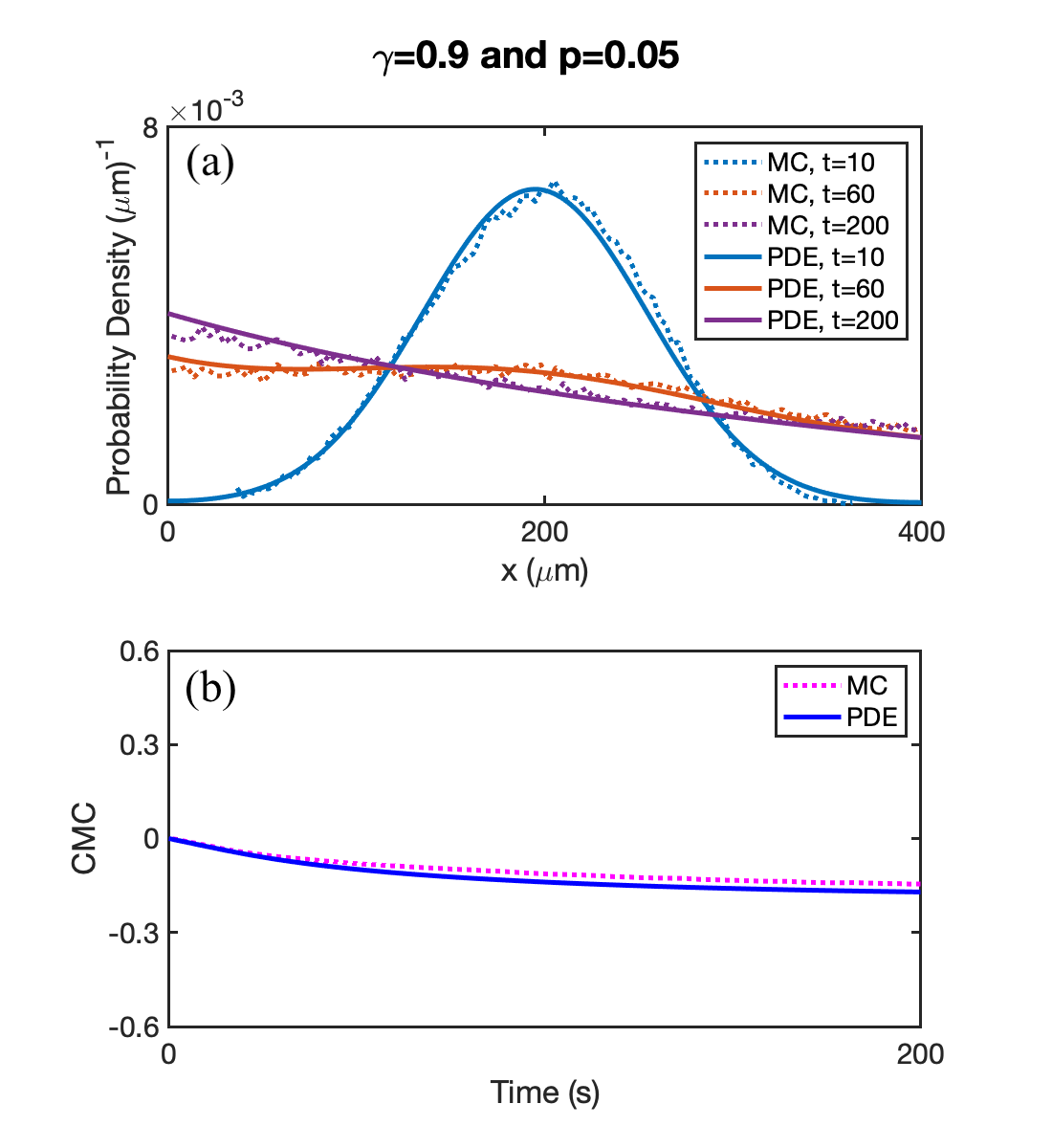}
	\end{minipage}
	\begin{minipage}[c][1\width]{
	   0.45\textwidth}
	   \centering
	   \includegraphics[width=1\textwidth]{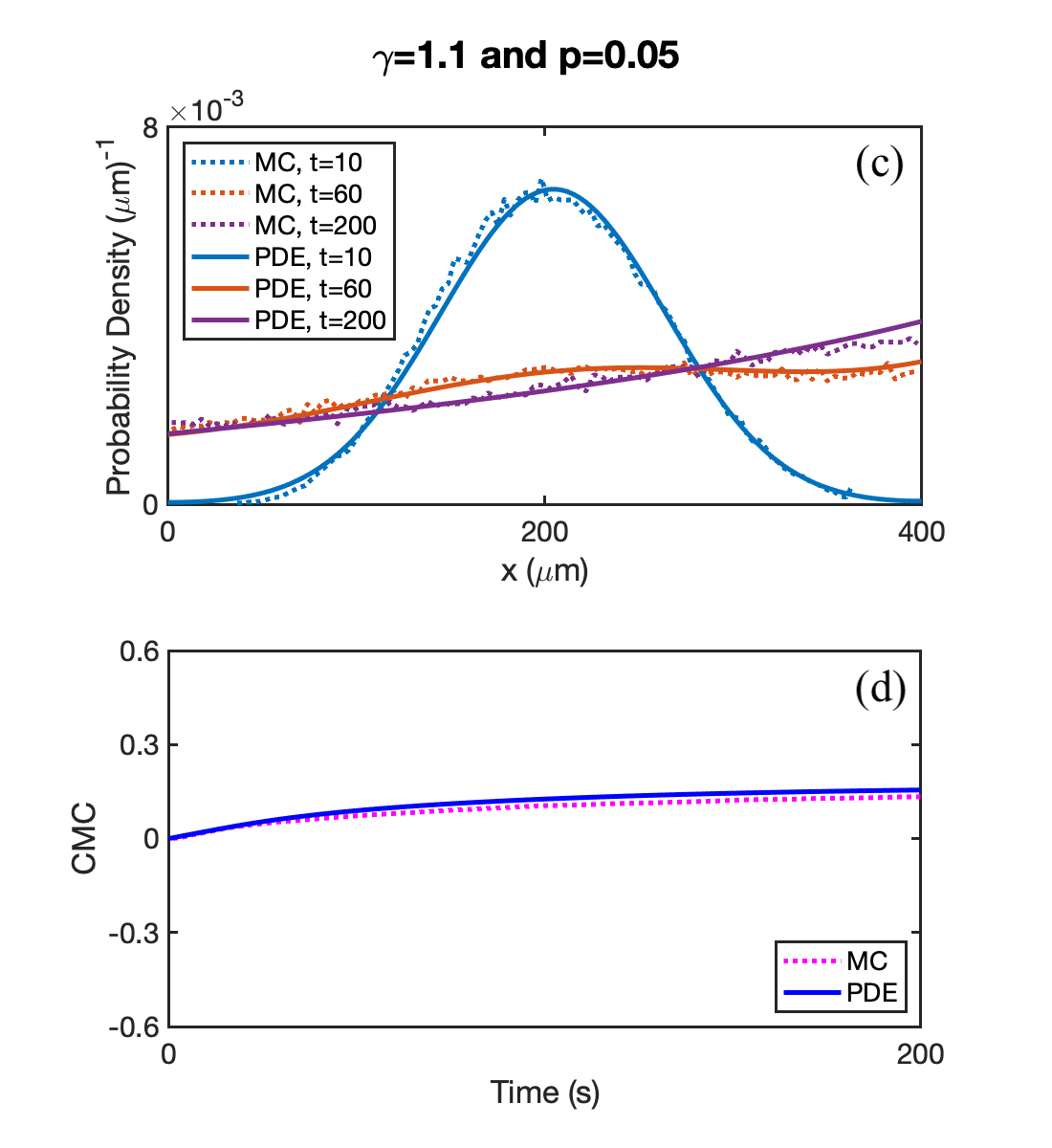}
	\end{minipage}
    \caption{(a) and (c): Comparisons of Monte-Carlo simulation and numerical solutions of  \eqref{main_equation_example} for two exponential gradients \eqref{exp_exp_gradients} at times $t=10, 60,200$ (sec) with $\gamma=0.9<\gamma^*$  and $\gamma=1.1>\gamma^*$, in which the snapshots move to the left and right, respectively. (b) and (d):
    Comparisons of the corresponding $\text{CMC}_x$.}
    \label{1d_exp_1p1_0p9}
\end{figure}

%%% Section 5
\section{Advection-diffusion equation for chemotaxis in response to two stimuli in a two-dimensional space}\label{AD:2dim}

In this section, we assume that the bacteria move in a two-dimensional space. Applying moment closure techniques and parabolic scaling \cite{erban2004individual,erban2005signal,xue2009multiscale,aminzare2013remarks, xue2015macroscopic}, 
we derive an
equation for the density of cells at the population level that  carries the description of an internal state of individuals in response to the extracellular signals. 

As introduced in Section \ref{Fokker-Planck}, let $p({\bf x}, \aa, \nu, \theta, t)$ be a density function that describes a population of agents at time $t$ and location ${\bf x}=(x,y)^\top$ with velocity $(\nu_1,\nu_2)=(\nu \cos \theta, \nu \sin \theta)$ and an internal state $\aa.$ For the sake of simplicity, by fixing a constant speed $\nu$, we let 
$p_{\theta}(x,y, \aa, t)$ denotes the density of bacteria centered at $(x,y)^\top$ which move to the direction $(\cos(\theta),\sin(\theta))^\top, \theta\in[0,2\pi)$, with the speed $\nu$.

 According to the forward Fokker-Planck  equation~(\ref{transport}), for $\theta \in [0, 2\pi)$, $p_{\theta} (x,y,\aa,t)$ satisfies
\begin{align}\label{transport_two_dim_general}
\begin{aligned}
\frac{\partial {p_{\theta}}}{\partial t}\;+&\;\;\dfrac{\partial }{\partial x} \left(\nu\cos(\theta){p_{\theta}}\right)\;+\;\;\frac{\partial}{\partial y} \left(\nu\sin(\theta){p_{\theta}}\right) \;+\; \frac{\partial }{\partial \aa} \left(f_{\theta}(\aa, S_{1}, S_{2})\; {p_{\theta}}\right)\\
= \; & \dfrac{1}{2\pi} \lambda(\aa, S_{1}, S_{2})\int_0^{2\pi} ({p_{\eta}}(x,y, \aa, t)-{p_{\theta}}(x,y, \aa,  t))\;d\eta, 
\end{aligned}
\end{align}
where $f_{\theta}$ and $\lambda$ describe the internal dynamics and tumbling rate, respectively. 

In the presence of two extracellular signals $S_{1}(x,y,t)$ and $S_{2}(x,y,t)$,
the evolution (\ref{deterministic_model}) of the internal state of the bacteria 
that move to the direction $(\cos(\theta),\sin(\theta))^\top$ with the speed $\nu$ is governed by the following ordinary differential equation. 
\begin{align}\label{internal dyn in 2d}
\begin{aligned}
    \frac{d \aa}{dt} &=    f_{\theta}(\aa ,S_{1}, S_{2})\\
    &= f_0(\aa ,S_{1}, S_{2})+\nu \cos(\theta) f_1^1(\aa ,S_{1}, S_{2})+\nu\sin(\theta)f_1^2(\aa ,S_{1}, S_{2}),
\end{aligned}
\end{align}
where the real-valued functions $f_{\theta}, f_{0}, f_{1}^{1}$, and $f_{1}^{2}$ are continuously differentiable.
We assume that $f_{0}, f_{1}^{1}$ and  $f_{1}^{2}$ have the Taylor expansions with respect to $\aa$ as follows:
\begin{align*}\label{Taylor_f}
    \begin{aligned}
    f_{0}&= A_{0} + A_{1}\aa +A_{2} \aa^{2} +\cdots,\\
    f_{1}^{1} &= B_{0}^{1}  + B_{1}^{1} \aa+ B_{2}^{1}\aa^{2}+ \cdots,\\
    f_{1}^{2} &= B_{0}^{2}  + B_{1}^{2} \aa+ B_{2}^{2}\aa^{2}+ \cdots.\\
    \end{aligned}
\end{align*}
Also, we assume that the tumbling rate $\lambda=\lambda(\aa, {S}_{1}, S_{2})$ has the Taylor expansion  
\[\lambda  = \alpha_{0} + \alpha_{1} \aa + \alpha_{2} \aa^{2} + \cdots. \] 
All the Taylor constants are functions of $S_1$ and $S_2$.

At a fixed time $t$, consider a population of bacteria with internal dynamics \eqref{internal dyn in 2d} and tumbling rate $\lambda$ that are located in $(x,y)$. We want to show that, under some conditions,  the population of bacteria,  which can be described by 
\begin{align*}\label{density_2d}
    n(x,y,t) = \int_{\mathbb{R}} \int_{0}^{2 \pi} p_{\theta}(x,y,\aa, t) d\theta d\aa,
\end{align*}
solves an advection-diffusion equation: 
\begin{equation*}
\frac{\partial n}{\partial t}\;=\; \dfrac{1}{2}\frac{\partial}{\partial x}\left(\frac{\nu^2}{\alpha_0}\frac{\partial n}{\partial x}-\frac{\nu^2\alpha_1B_0^1}{\alpha_0(A_1-\alpha_0)}\; n\right)+
\dfrac{1}{2}\frac{\partial}{\partial y}\left(\frac{\nu^2}{\alpha_0}\frac{\partial n}{\partial y}-\frac{\nu^2\alpha_1B_0^2}{\alpha_0(A_1-\alpha_0)}\; n\right).
\end{equation*}

Following the techniques from \cite{erban2004individual} and \cite{aminzare2013remarks}, we define 
the fluxes as
\begin{align*}
    \begin{aligned}
    j^{(1)} (x,y,t)& = \int_{\r} \int_{0}^{2\pi} \nu \cos(\theta)p_{\theta}(x,y,\aa,t) d\theta d\aa,\\
    j^{(2)} (x,y,t)& = \int_{\r} \int_{0}^{2\pi} \nu \sin(\theta) p_{\theta}(x,y,\aa,t) d\theta d\aa,
    \end{aligned}
\end{align*}
and the higher moments of the density and the fluxes as
\begin{align}\label{moments_2d}
    \begin{aligned}
    n_{i}(x,y,t)& = \int_{\r} \int_{0}^{2 \pi}\aa^{i} p_{\theta}(x,y,\aa,t) d \theta d\aa,&& \qquad i = 1,2,\ldots,\\
    j_{i}^{(1)}(x,y,t)&=\int_{\r} \int_{0}^{2 \pi}\aa^{i} \nu \cos(\theta)p_{\theta}(x,y,\aa,t) d\theta d\aa, &&\qquad i = 1,2, \ldots,\\
    j_{i}^{(2)}(x,y,t)&=\int_{\r} \int_{0}^{2 \pi}\aa^{i} \nu \sin(\theta)p_{\theta}(x,y,\aa,t) d\theta d\aa, &&\qquad i = 1,2, \ldots.
    \end{aligned}
\end{align}

\begin{assumption} \label{assumption:decay}
\normalfont
For any $\theta\in[0,2\pi)$, the density functions $p_{\theta}$ satisfy the decay condition \begin{equation}\label{decaying_cond}
    p_{\theta} (x,y,\aa,t) \leq C(x,y,t) e^{-c(x,y,t)\aa} 
\end{equation}
for some functions $C, c: \mathbb{R}^{2} \times [0, \infty) \rightarrow \mathbb{R}_{>0}.$ 
\end{assumption}

This assumption guarantees that the higher moments of the density and fluxes are well-defined. 

\begin{assumption} \label{assumption:moments}
For any $i \geq 2$, 
we assume that $n_i$, $j_{i}^{(1)}$, and $j_{i}^{(2)}$ are negligible compared to $n_{0},$ $j_{0}^{(1)},$ $j_{0}^{(2)}$, $n_{1},$ $j_{1}^{(1)}$ and $j_{1}^{(2)}$. 
\end{assumption}
This assumption is made for the purpose of more tractable calculations. 

\begin{assumption}\label{assumption:coeff}
$A_0=0$, $A_1\neq0$, $a_0\neq0$, $A_1\neq a_0$, $B_0^1\neq0$, and $B_0^2\neq0$.
\end{assumption}
This assumption guarantees the existence of unique solutions for some equations (see \eqref{eps_1st}-\eqref{eps_2nd} below). 

Multiplying (\ref{transport_two_dim_general}) by 1, $\nu \cos(\theta)$, $\nu \sin(\theta)$, and/or $\aa$ and integrating the resulting equations with respect to $\aa$ and $\theta$ over $\r$ and $[0, 2 \pi)$, respectively, we obtain the following six equations:
\begin{align}\label{2d:six eq}
    \begin{aligned}
    &\frac{\partial n}{\partial t} + \frac{\partial j^{(1)}}{\partial x} + \frac{\partial j^{(2)}}{\partial y}  =0,\\
    &\frac{\partial j^{(1)}}{\partial t} + \frac{\nu^{2}}{2} \frac{\partial n}{\partial x} = -\alpha_{0}j^{(1)} - \alpha_{1} j_{1}^{(1)}-\sum_{k\geq 2}\alpha_{k}j_{k}^{(1)},\\
    &\frac{\partial j^{(2)}}{\partial t} + \frac{\nu^{2}}{2} \frac{\partial n}{\partial y} = -\alpha_{0}j^{(2)} - \alpha_{1} j_{1}^{(2)}-\sum_{k\geq 2}\alpha_{k}j_{k}^{(2)},\\
   &\frac{\partial n_{1}}{\partial t} + \frac{\partial j_{1}^{(1)}}{\partial x} + \frac{\partial j_{1}^{(2)
    }}{\partial y} = A_{0}n + A_{1}n_{1} +    B_{0}^{1} j^{(1)} + B_{1}^{1} j_{1}^{(1)} +  B_{0}^{2}j^{(2)} + B_{1}^{2} j_{1}^{(2)} \\
    & \qquad \qquad \qquad \qquad \quad 
    +  \sum_{k \geq 2} A_{k}n_{k} + \sum_{k \geq 2} B_{k}^{1} j_{k}^{(1)} + \sum_{k \geq 2} B_{k}^{2} j_{k}^{(2)},\\
    &\frac{\partial j_{1}^{(1)}}{\partial t} + \frac{\nu^{2}}{2} \frac{\partial n_{1}}{\partial x}  = A_{0 } j^{(1)} + (A_{1}- \alpha_{0}) j_{1}^{(1)} + \frac{\nu^{2}}{2} B_{0}^{1} n + \frac{\nu^{2}}{2} B_{1}^{1}n_{1} 
   +\sum_{k \geq 2} (A_{k}- \alpha_{k-1}) j_{k}^{(1)} + \frac{\nu^{2}}{2} \sum_{k \geq 2} B_{k}^{1} n_{k} ,\\
    &\frac{\partial j_{1}^{(2)}}{\partial t} + \frac{\nu^{2}}{2} \frac{\partial n_{1}}{\partial y}  = A_{0 } j^{(2)} + (A_{1}- \alpha_{0}) j_{1}^{(2)} + \frac{\nu^{2}}{2} B_{0}^{2} n + \frac{\nu^{2}}{2} B_{1}^{2}n_{1}
  +\sum_{k \geq 2} (A_{k}- \alpha_{k-1}) j_{k}^{(2)} + \frac{\nu^{2}}{2} \sum_{k \geq 2} B_{k}^{2} n_{k} .
    \end{aligned}
\end{align}
Here, we used the decaying condition (\ref{decaying_cond}) which, for $i =0, 1, 2, \ldots$, yields  
\begin{align*}
    &\int_{\r} \int_{0}^{2 \pi} \aa^{i} \nu^{2}\cos(\theta) \sin(\theta) p_{\theta}(x,y, \aa, t) d \theta d\aa = 0,\\
    & \int_{\r} \int_{0}^{2\pi} \aa^{i} \nu^{2} \cos^{2}(\theta) p_{\theta}(x,y, \aa, t)  d\theta d\aa= \int_{\r} \int_{0}^{2 \pi} \aa^{i} \nu^{2} \sin^{2}(\theta) p_{\theta}(x,y,\aa, t) d\theta d\aa = \frac{\nu^{2}}{2} n_{i}(x,y,\aa,t),
\end{align*}
 where $n_{0}= n, j_{0}^{(1) }=j^{(1)}$ and $j_{0}^{(2)} = j^{(2)}$.

In what follows, we apply the parabolic scaling of space and time to the moment equation (\ref{moments_2d}), to derive a set of non-dimensional equations. 
Let $L, T, \nu_{0}$ and $N_{0}$ be scale factors for the length, time, velocity and the particle density, respectively. 
The parabolic scales of space and time are given by
\begin{align}\label{parabolic_space_time}
    \hat{x} = \Big( \frac{\varepsilon L}{\nu_{0} T}\Big) \frac{x}{L},  \quad   \hat{y} = \Big( \frac{\varepsilon L}{\nu_{0} T}\Big) \frac{y}{L}, \quad \hat{t} = \varepsilon^{2}\frac{t}{T},
\end{align}
for arbitrary small $\varepsilon>0$. Then, 
the dimensionless parameters are as follows. 
\begin{align}\label{dimensionless_parameters}
    \begin{aligned}
   &\qquad \qquad \hat{\nu} = \frac{\nu}{\nu_{0}},\quad \hat{n}= \frac{n}{N_{0}}, \quad \hat{j}^{(1)} = \frac{j^{(1)}}{N_{0} \nu_{0}}, \quad \hat{j}^{(2)}=\frac{j^{(2)}}{N_{0} \nu_{0}},\\
&    \hat{n_{i}}= \frac{n_{i}}{N_{0}}, \quad \hat{j_{i}}^{(1)} = \frac{j_{i}^{(1)}}{N_{0} \nu_{0}}, \quad \hat{j_{i}}^{(2)}=\frac{j_{i}^{(2)}}{N_{0} \nu_{0}}, \qquad \qquad \qquad i = 1,2,\ldots,\\
 &   \hat{\alpha}_{i} = T\alpha_{i}, \quad \hat{A}_{i} =T A_{i}, \quad \hat{B}_{i}^{1} = L B_{i}^{1}, \quad \hat{B}_{i}^{2} = L B_{i}^{2}, \qquad i = 1,2, \ldots.
    \end{aligned}
\end{align}
Denoting 
\begin{align*}
    \hat{w} = (\hat{n}, \hat{j}^{(1)}, \hat{j}^{(2)}, \hat{n}_{1}, \hat{j}_{1}^{(1)}, \hat{j}_{1}^{(2)})^{\top},
\end{align*}
and by Assumption \ref{assumption:moments},  we derive the following system of dimensionless moments from the dimensional equations (\ref{2d:six eq}):
\begin{align}\label{2d_dimensionless}
    \begin{aligned}
    \varepsilon^{2} \frac{\partial \hat{w}}{\partial \hat{t}} + \varepsilon \frac{\partial}{\partial \hat{x}} {\mathcal{P}}_{1} \hat{w}+ \varepsilon \frac{\partial}{\partial \hat{y}} {\mathcal{P}}_{2} \hat{w} = \varepsilon {\mathcal{Q}} \hat{w} + {\mathcal{R}} \hat{w}.
    \end{aligned}
\end{align}
Here, 
the matrices $\mathcal{P}_{1}, \mathcal{P}_{2},\mathcal{Q}$ and $\mathcal{R}$ are defined by partitioning into four $3\times3$ blocks such as
\begin{gather*}
    {\mathcal{P}}_{1} = 
    \begin{pmatrix}
   {P}_{1} & \mathbf{0}\\
    \mathbf{0} & {P}_{1} 
    \end{pmatrix}, \quad
{\mathcal{P}}_{2} = 
    \begin{pmatrix}
    {P}_{2} & \mathbf{0} \\
    \mathbf{0} & {P}_{2} 
    \end{pmatrix}, \quad
{\mathcal{Q}} = 
    \begin{pmatrix}
    \mathbf{0} & \mathbf{0} \\
    {Q}_{0} & {Q}_{1} 
    \end{pmatrix}, \quad
    {\mathcal{R}}= 
    \begin{pmatrix}
    {R}_{0} & {R}_{1} \\
    {S}_{0} & {S}_{1} +{R}_{0} 
    \end{pmatrix},
\end{gather*}
where $\mathbf{0}$ is a zero matrix of dimension $3$ and 
\begin{gather*}
   {P}_{1}= \begin{pmatrix}
    0 & 1 & 0\\
    \frac{\hat{\nu}^{2}}{2} & 0 & 0\\
    0 & 0 & 0
    \end{pmatrix},
    \hspace{0.2cm}
    {P}_{2} = \begin{pmatrix}
    0 & 0 & 1\\
    0 & 0 & 0\\
    \frac{\hat{\nu}^{2}}{2} & 0 & 0
    \end{pmatrix}, 
\hspace{0.2cm}
{Q}_{i} = 
\begingroup
\setlength\arraycolsep{2pt}
\begin{pmatrix}
    0 & \hat{B}_{i}^{1} & \hat{B}_{i}^{2}\\
    \frac{\hat{\nu}^{2}}{2}  \hat{B}_{i}^{1} & 0&0\\
    \frac{\hat{\nu}^{2}}{2}  \hat{B}_{i}^{2}  & 0 & 0
    \end{pmatrix}
    \endgroup,\\
      {R}_{i} =
      \begingroup
\setlength\arraycolsep{2pt}
      \begin{pmatrix}
    0 & 0 & 0\\
    0 & -\hat{\alpha}_{i}& 0\\
    0  & 0 & -\hat{\alpha}_{i}
    \end{pmatrix}
    \endgroup,  \hspace{0.2cm}
{S}_{i} =
\begingroup 
\setlength\arraycolsep{2pt}
\begin{pmatrix}
\hat{A}_{i} & 0 & 0\\
0 & \hat{A}_{i} &0\\
0 & 0 & \hat{A}_{i}
    \end{pmatrix}
    \endgroup, \qquad i=0,1.
\end{gather*}

To apply the regular perturbation method for $w,$ we set
\[ \hat{w} = \hat{w}^{0} + \varepsilon\hat{w}^{1} + \varepsilon^{2} \hat{w}^{2} + \cdots,\]
where 
\begin{align*}
    \hat{w}^{i} = \left(\hat{n}^{i}, \hat{j}^{(1)i}, \hat{j}^{(2)i}, \hat{n}_{1}^{i}, \hat{j}_{1}^{(1)i}, \hat{j}_{1}^{(2)i}\right)^\top, \qquad i=0,1,2,\ldots.
\end{align*}
Substituting $\hat{w}$ into the dimensionless moment system (\ref{2d_dimensionless})   
and collecting $\epsilon^i$ terms, for $i=0,1, 2$,
\begin{align}  
    &\varepsilon^{0} : \qquad {\mathcal{R}} \hat{w}^{0}=0 \label{eps_1st}\\
    &\varepsilon^{1} : \qquad {\mathcal{R}} \hat{w}^{1} = - {\mathcal{Q}} \hat{w}^{0}  + \frac{\partial}{\partial \hat x} {\mathcal{P}}_{1} \hat{w}^{0} + \frac{\partial }{\partial \hat y} {\mathcal{P}}_{2} \hat{w}^{0} \label{eps_2nd} \\
    & \varepsilon^{2} : \qquad {\mathcal{R}}\hat{w}^{2} = -{\mathcal{Q}} \hat{w}^{1} + \frac{\partial}{\partial \hat{x}} {\mathcal{P}}_{1}\hat{w}^{1} + \frac{\partial}{\partial \hat{y}} {\mathcal{P}}_{2} \hat{w}^{1} + \frac{\partial }{\partial \hat{t}} \hat{w}^{0}. \label{eps_3rd} 
\end{align}

By Assumption \ref{assumption:coeff},  (\ref{eps_1st})
has a unique solution $\hat{w}^{0}$ of the form
\begin{align*}
    \hat{w}^{0} = (\hat{n}^{0}, 0,0,0,0,0)^\top,
\end{align*}
where $\hat n^{0}$ is nonzero. 
The second equation (\ref{eps_2nd}) yields 
\begin{align}\label{2d:eq1}
    \begin{pmatrix}
    0\\
    -\hat{\alpha}_{0}\hat{j}^{(1)1} - \hat{\alpha}_{1} \hat{j}_{1}^{(1)1}\\
    -\hat{\alpha}_{0}\hat{j}^{(2)1} - \hat{\alpha}_{1}\hat{j}_{1}^{(2)1}\\
    \hat{A}_{1}\hat{n}_{1}^{1}\\
    (\hat{A}_{1} - \hat{\alpha}_{0})\hat{j}_{1}^{(1)1} + \frac{\hat{\nu}^{2}}{2} \hat{B}_{0}^{1} \hat{n}^{0}\\
    (\hat{A}_{1} - \hat{\alpha}_{0})\hat{j}_{1}^{(2)1} + \frac{\hat{\nu}^{2}}{2}\hat{B}_{0}^{2} \hat{n}^{0}
    \end{pmatrix}
     = \begin{pmatrix}
     0\\
     \frac{\hat{\nu}^{2}}{2} \frac{\partial \hat{n}^{0}}{\partial \hat{x}}\\
     \frac{\hat{\nu}^{2}}{2} \frac{\partial \hat{n}^{0}}{\partial \hat{y}}\\
     0\\ 0\\ 0
     \end{pmatrix};
\end{align}
and from the last two equalities of (\ref{2d:eq1}), it follows
\begin{align*}
    \hat{j}_{1}^{(1)1} = -  \frac{\hat{\nu}^{2} \hat{B}_{0}^{1} }{2 (\hat{A}_{1} - \hat{\alpha}_{0})}\hat{n}^{0}
    \qquad \text{and}
\qquad
\hat{j}_{1}^{(2)1} = -  \frac{\hat{\nu}^{2} \hat{B}_{0}^{2} }{2 (\hat{A}_{1} - \hat{\alpha}_{0})}\hat{n}^{0}.
\end{align*}
Moreover, plugging $\hat{j}_{1}^{(1)1}$ and $\hat{j}_{2}^{(2)1}$ into the second and third equalities in (\ref{2d:eq1}), we obtain 
\begin{align}\label{2d:eq3}
    \hat{j}_{1}^{(1)1} = - \frac{\hat{\nu}^{2}}{2 \hat{\alpha}_{0}}\frac{\partial \hat{n}^{0}}{\partial \hat{x}} + \frac{\hat{\nu}^{2}\hat{\alpha}_{1} \hat{B}_{0}^{1} }{2\hat{\alpha}_{0} (\hat{A}_{1} - \hat{\alpha}_{0})} \qquad \text{and}
\qquad
\hat{j}_{1}^{(2)1} = - \frac{\hat{\nu}^{2}}{2 \hat{\alpha}_{0}}\frac{\partial \hat{n}^{0}}{\partial \hat{y}} + \frac{\hat{\nu}^{2}\hat{\alpha}_{1} \hat{B}_{0}^{2} }{2\hat{\alpha}_{0} (\hat{A}_{1} - \hat{\alpha}_{0})}\; .
\end{align}
Noticing that the right hand side of (\ref{eps_3rd}) is in the image of ${\mathcal{R}}$ and $(1,0,0,0,0,0)^\top$ is in the kernel of ${\mathcal{R}},$ the right hand side of  (\ref{eps_3rd}) must be orthogonal to $(1,0,0,0,0,0)^\top$ by the Fredholm Alternative Theorem, which yields
\begin{align}\label{2d:eq4}
\frac{\partial }{\partial \hat{t}} \hat{n}^{0} +
    \frac{\partial }{\partial \hat{x}} \hat{j}^{(1)1} + \frac{\partial }{\partial \hat{y}}\hat{j}^{(2)1}=0.
\end{align}
By substituting the results in (\ref{2d:eq3}) into  (\ref{2d:eq4}), we  obtain the following equation for $\hat{n}^0$:
\begin{equation}\label{dimless_eq_2D}
\frac{\partial \hat{n}^{0}}{\partial \hat{t}}\;=\; \dfrac{1}{2}\frac{\partial}{\partial \hat{x}}\left(\frac{\hat{\nu}^2}{\hat{\alpha}_0}\frac{\partial \hat{n}^{0}}{\partial \hat{x}}-\frac{\hat{\nu}^2\hat{\alpha}_1 \hat{B}_0^1}{\hat{\alpha}_0(\hat{A}_1-\hat{\alpha}_0)}\; \hat{n}^{0}\right)+
\dfrac{1}{2}\frac{\partial}{\partial \hat{y}}\left(\frac{\hat{\nu}^2}{\hat{\alpha}_0}\frac{\partial \hat{n}^{0}}{\partial \hat{y}}-\frac{\hat{\nu}^2\hat{\alpha}_1 \hat{B}_0^2 }{\hat{\alpha}_0(\hat{A}_1-\hat{\alpha}_0)}\; \hat{n}^{0}\right).
\end{equation}
Similarly, we can  derive  the evolution equation for $\hat{n}^{1}$ which solves \eqref{dimless_eq_2D}.

For $n(x,y,t) = n^{0}(x,y,t) + \varepsilon n^{1}(x,y,t) + \mathcal{O}(\varepsilon^{2}),$ if the terms in $\mathcal{O}(\varepsilon^{2})$ are ignored,  (\ref{dimless_eq_2D}) for the original (dimensional) variable $n$ is transformed into
\begin{equation}\label{advection_diffusion_2D}
\frac{\partial n}{\partial t}\;=\; \dfrac{1}{2}\frac{\partial}{\partial x}\left(\frac{\nu^2}{\alpha_0}\frac{\partial n}{\partial x}-\frac{\nu^2\alpha_1B_0^1}{\alpha_0(A_1-\alpha_0)}\; n\right)+
\dfrac{1}{2}\frac{\partial}{\partial y}\left(\frac{\nu^2}{\alpha_0}\frac{\partial n}{\partial y}-\frac{\nu^2\alpha_1B_0^2}{\alpha_0(A_1-\alpha_0)}\; n\right).
\end{equation}

For the spatial domain $[0,L_{1}] \times [0, L_{2}]$, assuming that the  population of bacteria is conserved in time and there is no flux along the boundary, we can impose the following boundary conditions:
\begin{align}\label{Robin_general}
    \begin{cases}
    \vspace{0.1cm}
\;    D \dfrac{\partial n}{\partial x} (0,y,t) = \chi_{1}(0, y)n(0,y,t) \quad \mbox{and} \quad D \dfrac{\partial n}{\partial x}(L_{1},y,t) = \chi_{1}(L_{1},y) n(L_{1},y,t)\\
 \;   D \dfrac{\partial n}{\partial y} (x,0,t) = \chi_{2}(x, 0) n(x,0,t) \quad \mbox{and} \quad D \dfrac{\partial n}{\partial y}(x,L_{2},t) = \chi_{2}(x,L_{2}) n(x,L_{2},t),
    \end{cases}
\end{align}
where
   $ D = \frac{\nu^{2}}{2 \alpha_{0}} \; \mbox{and for $i=1,2,$} \; \chi_{i}(x,y) = \frac{\nu^{2} \alpha_{1} B_{0}^{i}}{2 \alpha_{0}(A_{1}-\alpha_{0})}.$

%****************
\subsection{Application to a population of \emph{E. coli} bacteria}\label{AD-application:2dim}

We now compute the coefficients of the macroscopic equation \eqref{advection_diffusion_2D} for \emph{E.\ coli} bacteria. 
As we discussed in Section \ref{internal:review}, in a two-dimensional space,  the internal state of \emph{E.\ coli} evolves according to the
following ODE system:
\begin{equation*}
\label{Example:internal:a:2D}
\frac{d\aa}{dt}\;=\; f_0(\aa,S_1,S_2) +   \nu \cos(\theta) f_{1}^{1}(\aa, S_{1}, S_{2}) + \nu \sin(\theta) f_{1}^{2}(\aa, S_{1}, S_{2}),
\end{equation*}
where 
\begin{align*}
    f_0 (\aa,S_1,S_2)&= \dfrac{\alpha}{\tau_a}N \aa(\aa-\aa_0)(\aa-1),\\
    f_{1}^{1}(\aa, S_{1}, S_{2}) &= N \aa (\aa -1) \Big(\dfrac{\gamma}{1+\gamma} \dfrac{\partial_{x}S_{1}}{S_{1} } + \dfrac{1}{1+\gamma}\dfrac{\partial_{x}S_{2}}{S_{2}} \Big) ,\\
    f_{1}^{2}(\aa,S_{1},S_{2}) &=N \aa (\aa -1) \Big(\dfrac{\gamma}{1+\gamma} \dfrac{\partial_{y}S_{1}}{S_{1} } + \dfrac{1}{1+\gamma}\dfrac{\partial_{y}S_{2}}{S_{2}} \Big).
\end{align*}
The  constant terms of the Taylor expansions of $f_{1}^{1}$ and $f_{1}^{2}$ are zero.  However, in Assumption \ref{assumption:coeff}, we saw that these constant terms must be non-zero. To fix this issue, 
we make a change of coordinate, $\hat{\aa} = \aa - \aa_{0},$ and obtain the following new internal dynamics of $\hat{\aa}:$
\begin{equation}
\label{Example:internal:hat:a:2D}
\frac{d\hat\aa}{dt}\;=\;  \hat f_0(\hat\aa,S_1,S_2) + \nu \cos (\theta) \hat f_1^{1}(\hat\aa,S_1,S_2) + \nu \sin(\theta) \hat f_{1}^{2}(\hat\aa, S_{1}, S_{2}),
\end{equation}
where
\begin{align*}\label{Example:internal:hat:a:2D:details}
\begin{aligned}
   \hat f_{0} (\hat \aa, S_1, S_2)& = p N \hat \aa (\hat \aa + q) (\hat \aa + q -1),\\
    \hat f_{1}^{1}(\hat \aa, S_{1}, S_{2}) &= N (\hat \aa + q) (\hat \aa + q -1) \Big(\dfrac{\gamma}{1+\gamma} \dfrac{\partial_{x}S_{1}}{S_{1} } + \dfrac{1}{1+\gamma}\dfrac{\partial_{x}S_{2}}{S_{2}} \Big) ,\\
    \hat f_{1}^{2}(\hat \aa,S_{1},S_{2}) &=N (\hat \aa + q) (\hat \aa + q -1) \Big(\dfrac{\gamma}{1+\gamma} \dfrac{\partial_{y}S_{1}}{S_{1} } + \dfrac{1}{1+\gamma}\dfrac{\partial_{y}S_{2}}{S_{2}} \Big). 
    \end{aligned}
\end{align*}
We let 
\begin{equation}\label{Example:translation}
q=\aa_0\quad \mbox{and}\quad p=\dfrac{\alpha}{\tau_a}
\end{equation}
 represent the adapted value and the the speed of adaptation, respectively. 

Next, we transform the tumbling rate, discussed in \eqref{tumble_rate}, into the new coordinate $\hat\aa$ as follows:
 \begin{equation}
 \label{Exampletumbling:hat:2D}
 \lambda(\hat\aa)\;=\;\lambda_0+r (\hat\aa+q)^H,
\qquad \mbox{where} \quad r=\dfrac{1}{\tau \aa_{0}^{H}}.
\end{equation}

All the model parameters $N, p, q, r,$ and $H$ are assumed to be positive constants and are as given in Table \ref{Table:Parameters}.

Let Assumption \ref{assumption:decay} hold.
In the following two lemmas, we provide sufficient conditions that lead to Assumption \ref{assumption:moments}. 

\begin{lemma}[Shallow condition]\label{shallow_condition_example_2d}
Let $c = \min\{q, 1-q\}$. If for any $(x, y)\in[0,L_1]\times[0,L_2]$, and any $\theta\in[0,2\pi)$
\begin{align}\label{Example:shallow:condition:2D}
    \Big| \cos (\theta)\Big(\dfrac{\gamma}{1+\gamma} \dfrac{\partial_{x}S_{1}}{S_{1} } + \dfrac{1}{1+\gamma}\dfrac{\partial_{x}S_{2}}{S_{2}} \Big)  + \sin (\theta) \Big(\dfrac{1}{1+\gamma} \dfrac{\partial_{y}S_{1}}{S_{1} } + \dfrac{1}{1+\gamma}\dfrac{\partial_{y}S_{2}}{S_{2}} \Big)\Big| \leq \frac{c p}{\nu},
\end{align}
and  $|\hat \aa (0)| \leq c$,
then $|\hat \aa(t)| \leq c$ for all $t \geq 0$.
\end{lemma}

\begin{proof}
To show $|\hat \aa(t)| \leq c$, it suffices to show that ${d\hat a}/{dt}<0$ (respectively, $>0$) at $\hat a=c$ (respectively, $\hat a=-c$). 
Using \eqref{Example:shallow:condition:2D} and $c<1-q$ into  \eqref{Example:internal:hat:a:2D}, we obtain the desired result. 
\end{proof}

Note that the inequality \eqref{Example:shallow:condition:2D} holds if either the adaptation rate $p$ is sufficiently large or $\gamma$, $S_1$, and $S_2$ are chosen so that the LHS of \eqref{Example:shallow:condition:2D} is sufficiently small. See the examples given in Section \ref{simulation:2dim} for more details. 

\begin{lemma}\label{moment_closure}
Let $L, T, \nu_{0},$ and $N_{0}$ be scale factors for the length, time, velocity, and cell density, respectively, as introduced in  \eqref{parabolic_space_time} and \eqref{dimensionless_parameters}. We define dimensionless quantities as follows:
\begin{align*}
    \widehat{\Big(\dfrac{\nabla_{\bf{x}}S_{i}}{S_{i}}\Big)} = \dfrac{\nu_{0}}{\varepsilon} \dfrac{\nabla_{\bf{x}}S_{i}}{S_{i}}, (i=1,2), \quad \hat{N} = TN, \quad \hat{\gamma} = \gamma, \quad \hat{p}=p, \quad \hat{q} = q \quad  \mbox{and} \quad \hat{r} = Tr.
\end{align*}
Then, under the shallow condition \eqref{Example:shallow:condition:2D}, for any $i \geq 1$,
\begin{align*}
    \dfrac{\hat{j}_{i}^{(1)}}{\hat{n}} \leq \; \mathcal{C}^{(1)}_{i} \varepsilon^{i}, \quad
    \dfrac{\hat{j}_{i}^{(2)}}{\hat{n}}  \leq \; \mathcal{C}^{(2)}_{i} \varepsilon^{i}, \quad
    \mbox{and} \quad
    \dfrac{\hat{n}_{i}}{\hat{n}} \leq \; \mathcal{D}_{i} \varepsilon^{i},
\end{align*}
for some constants $\mathcal{C}^{(1)}_{i} = O(1),$ $\mathcal{C}^{(2)}_{i}=O(1)$, and $\mathcal{D}_{i}=O(1)$.
\end{lemma}
The proof of Lemma \ref{moment_closure} can be completed as proved in \cite{aminzare2013remarks, xue2015macroscopic}; thus we omit the proof.

The shallow condition \eqref{Example:shallow:condition:2D}
guarantees that
the higher moments $n_{i},j_{i}^{(1)}$, and $j_{i}^{(2)}$, $i\geq 2$, 
are of order $\varepsilon^{2}$,  $O(\varepsilon^{2})$.
Indeed, we can close the moment equations \eqref{2d:six eq} by considering the higher moments as the error terms of $O(\varepsilon^{2}).$

Simple calculations show that the Taylor coefficients of $\hat f_0$, $\hat f_1^1$, and $\hat f_1^2$ are   
\begin{gather*}
    A_{0} = 0, \quad A_{1} = Npq (q-1), \quad
    B_{0}^{1} = N q (q-1) \Big ( \dfrac{\gamma}{1+\gamma} \dfrac{\partial_{x}S_{1}}{S_{1}} + \dfrac{1}{1+\gamma} \dfrac{\partial_{x} S_{2}}{S_{2}}  \Big),\\
    B_{0}^{2} = N q (q-1) \Big(\dfrac{\gamma}{1+\gamma}\dfrac{\partial_{y}S_{1}}{S_{1}} + \dfrac{1}{1+\gamma} \dfrac{\partial_{y} S_{2}}{S_{2}}  \Big),\quad
    \alpha_{0} =\lambda_{0} +  r q^{H},\quad \alpha_{1} = \dfrac{r H q^{H}}{q}\;, 
\end{gather*}
that satisfy Assumption \ref{assumption:coeff}. 
Then, with Assumption \ref{assumption:decay}, the shallow condition for the stimuli, and the
internal dynamics \eqref{Example:internal:hat:a:2D}, 
a population of \emph{E. coli}, $n(x,y,t)$, solves the following equation
\begin{align}\label{main_equation_example_2d}
    \dfrac{\partial n}{\partial t}  = \nabla \cdot \Big( D \nabla n - \chi \; \Big(\dfrac{\gamma}{1+\gamma} \dfrac{\nabla_{\bf x}S_{1}}{S_{1}} + \dfrac{1}{1+\gamma} \dfrac{\nabla_{\bf x} S_{2}}{S_{2}} \Big) \; n \Big),
\end{align}
where the diffusion coefficient $D$ and the advection constant $\chi$ are 
\begin{align*}
    D = \frac{\nu^{2}}{2(\lambda_{0} + r q^{H})} >0,
 \qquad \chi = \frac{r NH q^{H}(q-1)\nu^{2}}{2(\lambda_{0} + rq^{H})(Npq(q-1) - \lambda_{0} - r q^{H})} >0\;.
\end{align*}

For the spatial domain $[0, L_{1}] \times [0, L_{2}]$, the boundary conditions \eqref{Robin_general} become
\begin{align}\label{Robin_condition_2D}
    \begin{cases}
       \vspace{0.1cm}
\;    D \dfrac{\partial n}{\partial x} (0,y,t) = \chi V_{1}(0,y)n(0,y,t) \quad \mbox{and} \quad D \dfrac{\partial n}{\partial x} (L_{1},y,t) = \chi V_{1}(L_{1},y)n(L_{1},y,t)  , \\
\; D \dfrac{\partial n}{\partial y} (x,0,t) =\chi V_{2}(x,0)n(x,0,t) \quad \mbox{and} \quad
    D \dfrac{\partial n}{\partial y} (x,L_{2},t) =\chi V_{2}(x,L_{2})n(x,L_{2},t),
    \end{cases}
\end{align}
where
\begin{align*}
      V_{1}(x,y)= \dfrac{\gamma}{1+\gamma} \dfrac{\partial_{x} S_{1}}{S_{1}} + \dfrac{1}{1+\gamma} \dfrac{\partial_{x}S_{2}}{S_{2}} \quad \mbox{and} \quad     V_{2}(x,y)= \dfrac{\gamma}{1+\gamma} \dfrac{\partial_{y} S_{1}}{S_{1}} + \dfrac{1}{1+\gamma} \dfrac{\partial_{y}S_{2}}{S_{2}}.
\end{align*}

In the following lemma, we present sufficient  conditions that guarantee  the existence and uniqueness of solutions of \eqref{main_equation_example_2d} with the boundary condition \eqref{Robin_condition_2D}. 
The proof is followed by Lemma 1 due to the assumptions on $V_{1}$ and $V_{2}$.

\begin{lemma}\label{lem:existence 2D}
Let $V_{1}(x,y)$ and $V_{2}(x,y)$ be continuous on $\Omega:=[0,L_{1}]\times [0,L_{2}]$, 
and $n_{0}(x,y)$ be a smooth non-negative function. 
If $V_{1}(x,y) = V_{1}(x)$ and $V_{2}(x,y)=V_{2}(y),$ then  \eqref{main_equation_example_2d} with the boundary condition \eqref{Robin_condition_2D} and the initial condition $n(x,y,0) =n_{0}(x,y)$ admits a unique solution $n(x,y,t)$ in the form of $ \sum_{n=1}^{\infty} X_{n}(x)Y_{n}(y) T_{n}(t)$. Moreover, $n(x,y,t)$ is uniformly bounded in $x,y$ and $t.$
\end{lemma}

%****************
\subsection{Steady state solution of advection-diffusion equation with zero flux boundary conditions}\label{steady-state-2D}

In a similar way to explaining the direction of bacterial migration in Section \ref{steady-state-1D}, we explore properties of the steady state solution of the advection-diffusion equation \eqref{main_equation_example_2d} and predict the direction of bacteria. To do this, we choose $S_{1}$, $S_{2}$ and $\gamma$ so
that $V_{1}(x,y)$ and $V_{2}(x,y)$ satisfy the conditions given in Lemma \ref{lem:existence 2D}. 

To compute the steady state solution of the advection-diffusion equation \eqref{main_equation_example_2d} with zero flux boundary conditions \eqref{Robin_condition_2D}, we let the flux at $x$ direction and the flux at $y$ direction be zero, i.e., 
\begin{align*}
   J_x(x,y)&:= D\dfrac{\partial n}{\partial x} - \chi V_1(x,y)n=0, \\
    J_y(x,y)&:=D\dfrac{\partial n}{\partial y} - \chi V_2(x,y)n=0, 
\end{align*}
which yield
\begin{align}\label{grad-log-n}
\begin{pmatrix}
\frac{\partial }{\partial x} \log n\\
\frac{\partial }{\partial y} \log n 
\end{pmatrix}
=
\dfrac{\chi}{D}
\begin{pmatrix}
V_1(x,y)\\
V_2(x,y)
\end{pmatrix}.
\end{align}
Note that this equation cannot be satisfied for any arbitrary $V_1$ and $V_2$.  Since the LHS is a gradient, a necessary and sufficient condition for the equation to hold is 
\begin{equation}\label{condition-V1-V2}
\dfrac{\partial V_1}{\partial y}=
\frac{\partial V_2}{\partial x}\;.
\end{equation}
Note that if $V_1$ and $V_2$ satisfy the conditions of Lemma \ref{lem:existence 2D}, they automatically satisfy \eqref{condition-V1-V2}.
Under this condition, the steady state solution can be obtained by simple integration of  \eqref{grad-log-n}:
\begin{equation}
    \Phi(x,y) = \Phi({c}_1,{c}_2) \exp\left\{\dfrac{\chi}{D}\left(\displaystyle\int_{{c}_1}^x V_1(z,y)\;dz + \int_{{c}_2}^y V_2({c}_1,z)\;dz\right)\right\},
\end{equation}
where $({c}_1, {c}_2)\in[0,L_1]\times [0,L_2]$ are chosen such that  $\Phi({c}_1,{c}_2)$ is a positive constant. 
Similar to what we discussed in Section \ref{steady-state-1D}, if $\partial V_1/\partial x\leq0$ and $\partial V_2/\partial y\leq0$, then  the signs of $V_1$ and $V_2$ at the initial point $(x_0,y_0)$ can determine the direction of the motion of bacteria. We let $\gamma_1^*$ be the bifurcation value that determines the right/left direction (i.e., $V_1(x_0,y_0,\gamma_1^*)=0$) and $\gamma_2^*$ be the bifurcation value that determines the up/down direction (i.e., $V_2(x_0,y_0,\gamma_2^*)=0$). Then, three scenarios are possible: 
(i) for $\max\{\gamma^*_1,\gamma^*_2\}<\gamma$, the bacteria move to the northeast and accumulate in  $A_1:=\{ x_0<x<L_1, y_0<y<L_2\}$; 
(ii) for $\min\{\gamma^*_1,\gamma^*_2\}<\gamma<\max\{\gamma^*_1,\gamma^*_2\}$ the bacteria either move to the southeast and accumulate in  $A_4:=\{x_0<x<L_1, 0<y<y_0\}$ or move to the northwest and accumulate in  $A_2:=\{0<x<x_0, y_0<y<L_2\}$;
(iii) for $\gamma<\min\{\gamma^*_1,\gamma^*_2\}$, the bacteria move to the southwest and accumulate in  $A_3:=\{0<x<x_0, 0<y<y_0\}$. 

In the following section, we consider three sets of stimuli, which their corresponding $V_1$ and $V_2$ satisfy the condition of Lemma \ref{lem:existence 2D} (and hence   \eqref{condition-V1-V2}). For each set we find the bifurcation values which determine the direction of bacteria.

%%% Section 6
\section{Monte-Carlo agent-based simulations in two-dimensional space}\label{simulation:2dim}

To validate our two-dimensional macroscopic approximation \eqref{main_equation_example_2d}, we run a Monte-Carlo simulation for microscopic equation  \eqref{transport_two_dim_general}. 
Our numerical experimental set up is very similar to that of Section \ref{simulation:1dim}, which we generalize to a two-dimensional space as follows. Note that since this work is motivated by  \cite{kalinin2010responses}, we choose a computational setting to be qualitatively similar to the experimental set up of \cite{kalinin2010responses} as well.

\begin{description}[leftmargin=*]
\item[Spatial Domain.] A channel of area of $400 \mu m$ by $1600 \mu m$ ($x \in [0,400],$ $y \in [0,1600]$). 
\item[Stimuli.]
Along the two sides of the channel $x=0$ and $x=400$,
two opposing chemical signals $S_{1}(x,y)$  and $S_{2}(x,y)$, which respectively represent the concentrations of MeAsp and serin at $(x,y)$, 
flow and diffuse across the channel.
Three  sets of stimuli will be considered in Sections \ref{lin_ligands:2D}-- \ref{mix_ligands:2D}, below. 

\item[Initial Condition.] At $t=0$ (sec), an ensemble of 100,000 agents is located in the center of the channel ($x=200$ and $y=800$). 

\item[Boundary Condition.] 
We use reflecting boundary conditions at $x=0,400$ and $y=0,1600$ so the cells stay in the domain for all time.

\item[Simulation Duration.]
We simulate the bacterial behavior for $t \in [0,200]$. In Sections \ref{lin_ligands:2D} and \ref{exp_ligands:2D} we observed that the solutions of the Monte-Carlo simulation and the numerical solutions of \eqref{main_equation_example_2d} become stationary at $t=200$.
\end{description}

The distributions of the solutions are displayed by using histograms with 2500 equal-sized bins. 
To solve the advection-diffusion equation \eqref{main_equation_example_2d} with  boundary conditions \eqref{Robin_condition_2D}, we use an explicit finite difference method.
The summary of input data is given in Table \ref{Table:Simulation} (see Appendix \ref{Appendix2}), and more details can be also found in Section \ref{simulation:1dim}.

In what follows, we show some numerical results for three different choices of the stimuli combinations: Linear--Linear in Section \ref{lin_ligands:2D}, 
Exponential--Exponential in Section \ref{exp_ligands:2D}, and
Linear $\times$ Exponential--Linear $\times$ Exponential in Section \ref{mix_ligands:2D}. 
We will show that (i) for some $\gamma^{*}$, when 
$\gamma>\gamma^*$, the bacteria move to the the gradient of increasing MeAsp and when $\gamma<\gamma^*$, the bacteria move to the gradient of increasing serine; and (ii) under the condition of Lemma \ref{shallow_condition_example_2d}, 
the Monte-Carlo agent-based simulations and the numerical  solutions of \eqref{main_equation_example_2d} agree well.  

%*************
\subsection{Chemotaxis in response to two linear gradients}\label{lin_ligands:2D}

Let 
$S_{1}(x,y)=0.5 x + 130$ and $S_{2}(x,y)=-0.03x + 20$
be two opposing linear gradients for MeAsp and serine, respectively. 
Note that the stimuli are constant with respect to $y$. In this case, $V_1(x,y)= V(x)$, as defined in Section \ref{lin_ligands:1D}, and $V_2(x,y)=0$. Therefore, the condition of Lemma \ref{lem:existence 2D} and hence \eqref{condition-V1-V2} hold and the bacteria only move to the right or left (no up or down movement). Furthermore, the bifurcation value is equal to $\gamma^*\approx 0.985$, as computed in  Section \ref{lin_ligands:1D}. 

For the given linear gradients, Figures \ref{2d_lin_1p1}(a, b) (respectively, Figures \ref{2d_lin_0p9}(a, b)) display the distributions of the normalized density of bacteria obtained from the Monte-Carlo agent-based simulation and numerical simulation of \eqref{main_equation_example_2d} for $\gamma=1.5$ (respectively, $\gamma=0.5$). 
The simulations are shown in three snapshots at times $t=0$ (left), $t=60$ (middle), and $t=200$ (right).
Figure \ref{2d_lin_1p1}(c)(respectively, Figure \ref{2d_lin_0p9}(c)) displays the corresponding CMCs in $x-$direction and $y-$direction. 
%******Fig 5******
\begin{figure}[ht!]
  \centering
\subfloat[Monte-Carlo simulation for $\gamma=1.5$ and $p=1$]{\includegraphics[width=\textwidth]{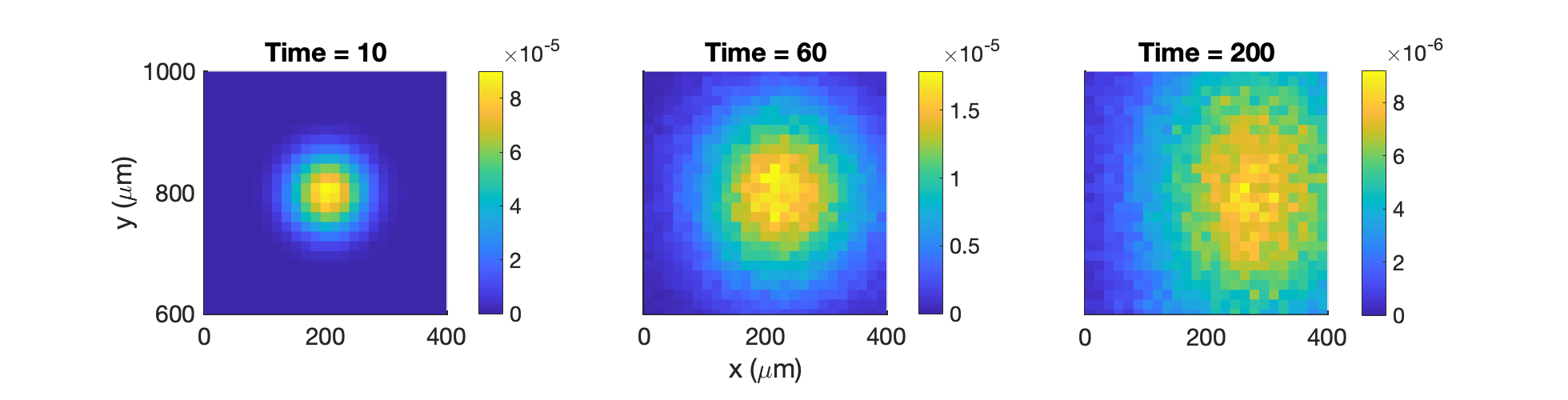}}\\
\vspace{-0.35cm}
\subfloat[Numerical solutions of  \eqref{main_equation_example_2d} for $\gamma=1.5$ and $p=1$]{\includegraphics[width=\textwidth]{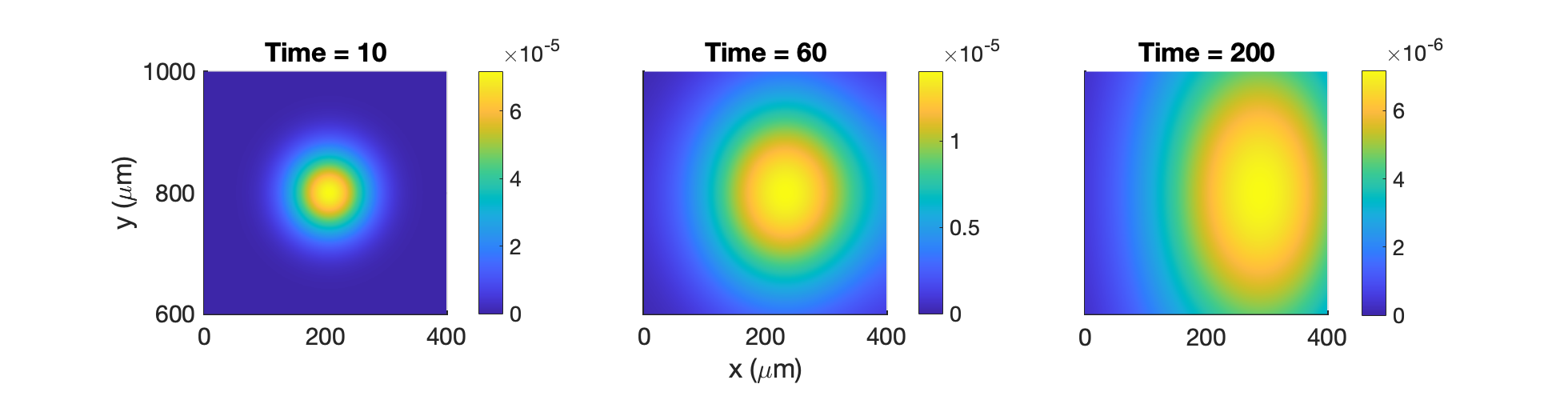}}\\
\vspace{-0.4cm}
\subfloat[$\text{CMC}_x$ (left) and $\text{CMC}_y$ (right)  ]{\includegraphics[width=0.8\textwidth]{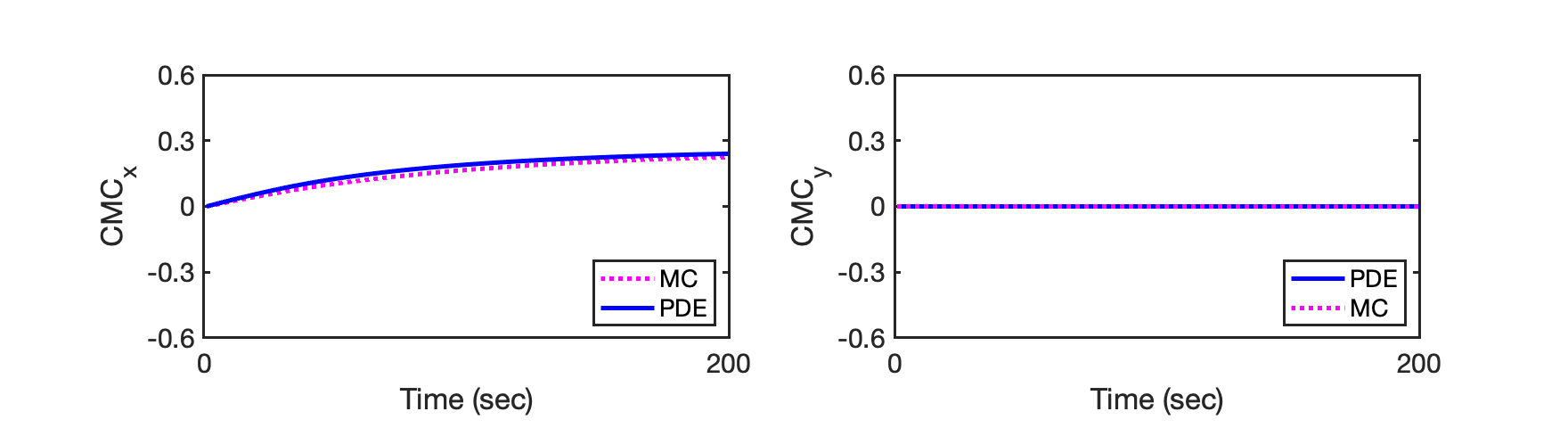}}
\caption{
	(a) and (b): Comparisons of the Monte-Carlo simulations and numerical solutions of  \eqref{main_equation_example_2d} in response to two linear gradients, when $\gamma =1.5$. In this case the bacteria move to the right,  the  gradient of increasing MeAsp.  (c): Comparisons of the corresponding CMCs.}
\label{2d_lin_1p1}
\end{figure}
%******Fig 6******
 \begin{figure}[ht!]
  \centering
\subfloat[Monte-Carlo simulation for $\gamma=0.5$ and $p=1$]{\includegraphics[width=\textwidth]{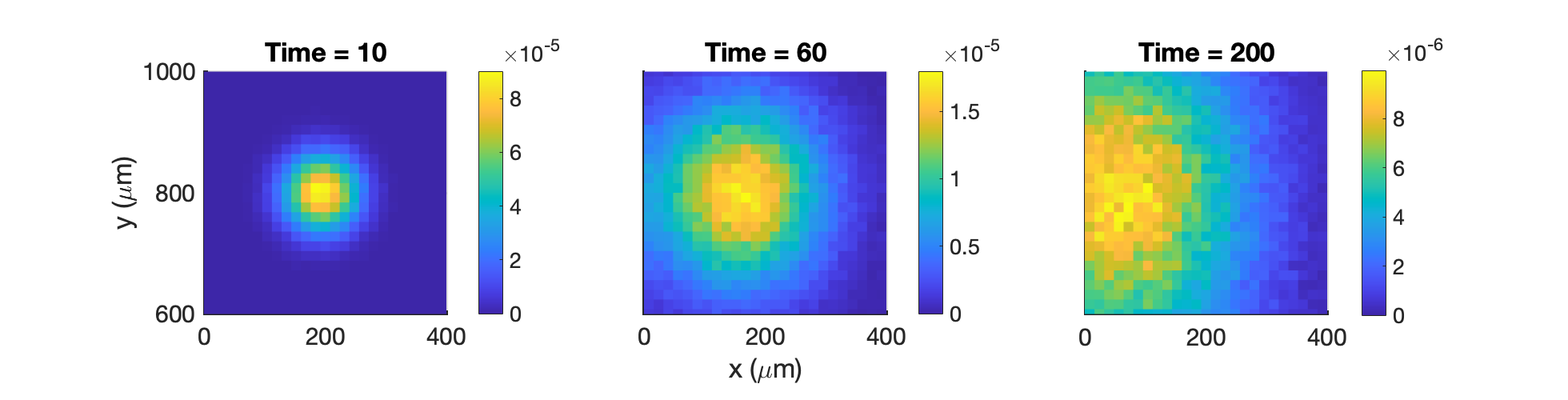}}\\
\vspace{-0.35cm}
\subfloat[Numerical solutions of  \eqref{main_equation_example_2d} for $\gamma=0.5$ and $p=1$]{\includegraphics[width=\textwidth]{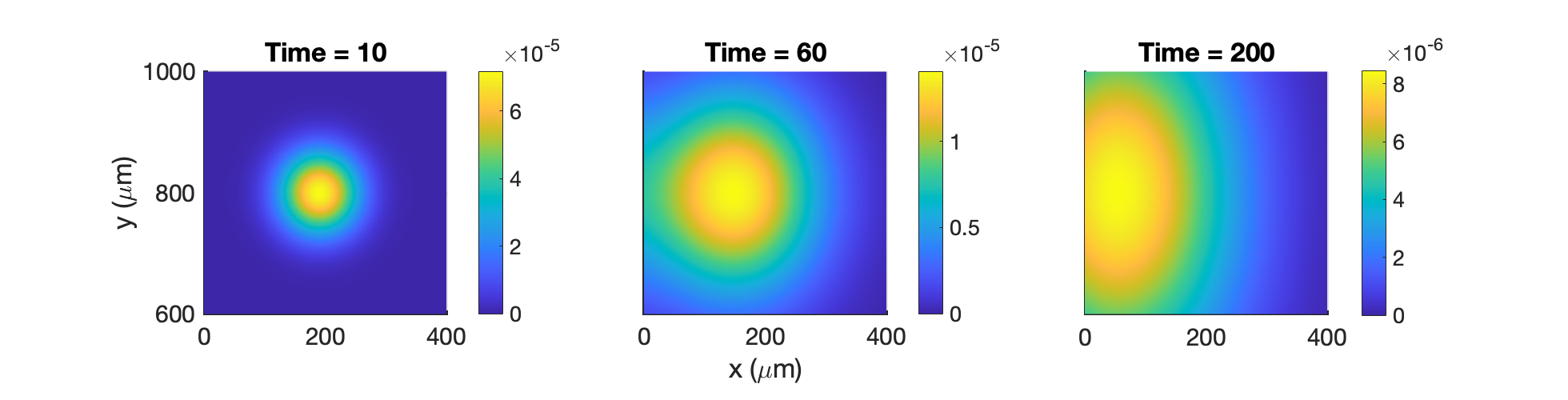}}\\
\vspace{-0.4cm}
\subfloat[$\text{CMC}_x$ (left) and $\text{CMC}_y$ (right)  ]{\includegraphics[width=0.8\textwidth]{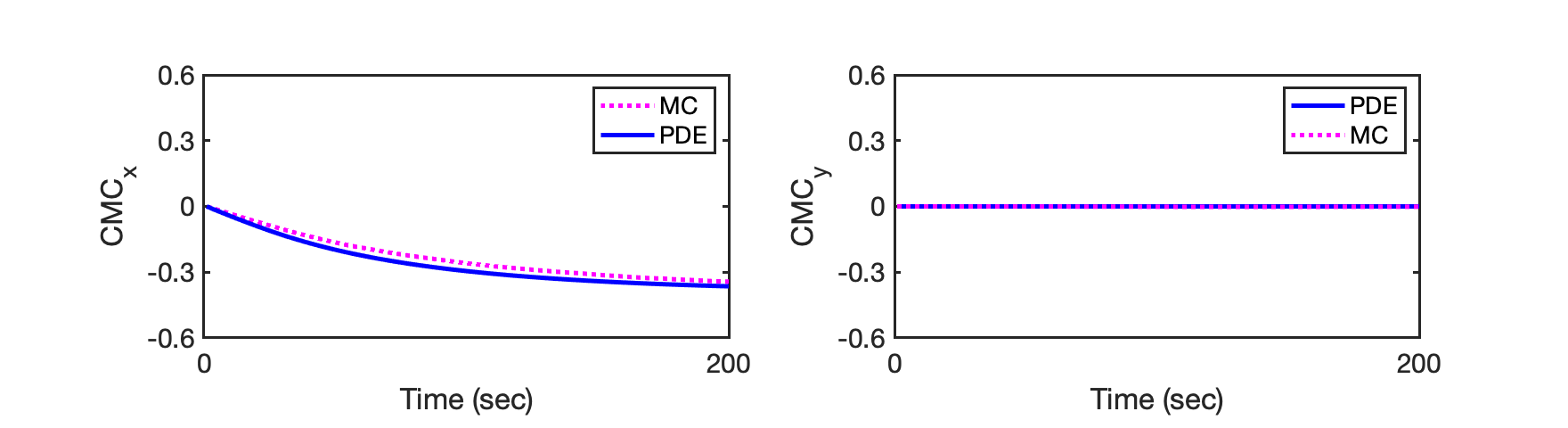}}
\caption{(a) and (b): Comparisons of the Monte-Carlo simulations and numerical solutions of  \eqref{main_equation_example_2d} for two linear gradients in  \eqref{linear_linear_gradients},  $\gamma =0.5$, and $p=1$.  In this case the bacteria move to the left,  the  gradient of increasing serine.
  Plots in (a) and (b) are displayed only for $(x,y) \in [0, 400] \times [600, 1000]$. (c): Comparisons of the corresponding CMCs.}
\label{2d_lin_0p9}
\end{figure}
 
 In Figures \ref{2d_lin_1p1} and \ref{2d_lin_0p9}, 
 the numerical solutions of  \eqref{main_equation_example_2d} are in good agreement with the results of the agent-based simulation. The snapshots of the distribution move to the gradient of increasing MeAsp in Figure \ref{2d_lin_1p1} or serine in Figure \ref{2d_lin_0p9}.
 Recalling the bifurcation value of $\gamma^{*} \approx 0.985$ in Section \ref{lin_ligands:1D}, these figures confirm that the chemotactic preference of bacteria depends on the relative abundances of receptors, i.e., when $\gamma=1.1>\gamma^*$, the bacteria move to the  gradient of increasing MeAsp ($\text{CMC}_x>0$ and increasing) and when $\gamma=0.9<\gamma^*$, the bacteria move to the  gradient of increasing serine ($\text{CMC}_x<0$ and decreasing). 
 Note that these numerical examples qualitatively reproduce the bacterial behaviors observed in
 \cite{kalinin2010responses}. 

Since $S_{1}$ and $S_{2}$ are independent of $y$, the bacteria move in the $y$-direction very slightly, as evidenced by $\text{CMC}_y\approx0$. Thus, although we run all the simulations on the domain $[0,400] \times [0, 1600]$, we  display a smaller domain,  $[0,400] \times [600, 1000]$.

%*************
\subsection{Chemotaxis in response to two exponential gradients}\label{exp_ligands:2D}

We assume that bacteria are exposed to two opposing exponential gradients 
\[S_{1}(x,y) = 130 e^{0.0023 x}\quad  \text{and} \quad S_{2}(x,y) = 8 e^{-0.0023(x-400)}.\]
In this case, $V_1(x,y)= V(x)$, as defined in Section \ref{exp_ligands:1D}, and $V_2(x,y)=0$. Therefore, condition \eqref{condition-V1-V2} holds and the bacteria only move to the right or left (no up or down movement). Furthermore, the bifurcation value is equal to $\gamma^*=1$, as computed in  Section \ref{exp_ligands:1D}. 

In Figures \ref{2d_exp_1p1} and \ref{2d_exp_0p9}, we compare the results of the Monte-Carlo simulation with numerical solution of  \eqref{main_equation_example_2d} and their corresponding CMCs. From the plots, we can see that \eqref{main_equation_example_2d} captures the behavior of individuals well. Recalling the bifurcation value $\gamma^{*} =1$ of the ratio of Tar to Tsr in Section \ref{exp_ligands:1D}, as expected, the individuals travel to the right when $\gamma=1.1>\gamma^*$ as in Figure \ref{2d_exp_1p1} and move to the left when $\gamma=0.9<\gamma^*$ as in Figure \ref{2d_exp_0p9}.

%******Fig 7******
\begin{figure}[H]
  \centering
\subfloat[Monte-Carlo simulation for $\gamma=1.1$ and $p=0.1$]{\includegraphics[width=\textwidth]{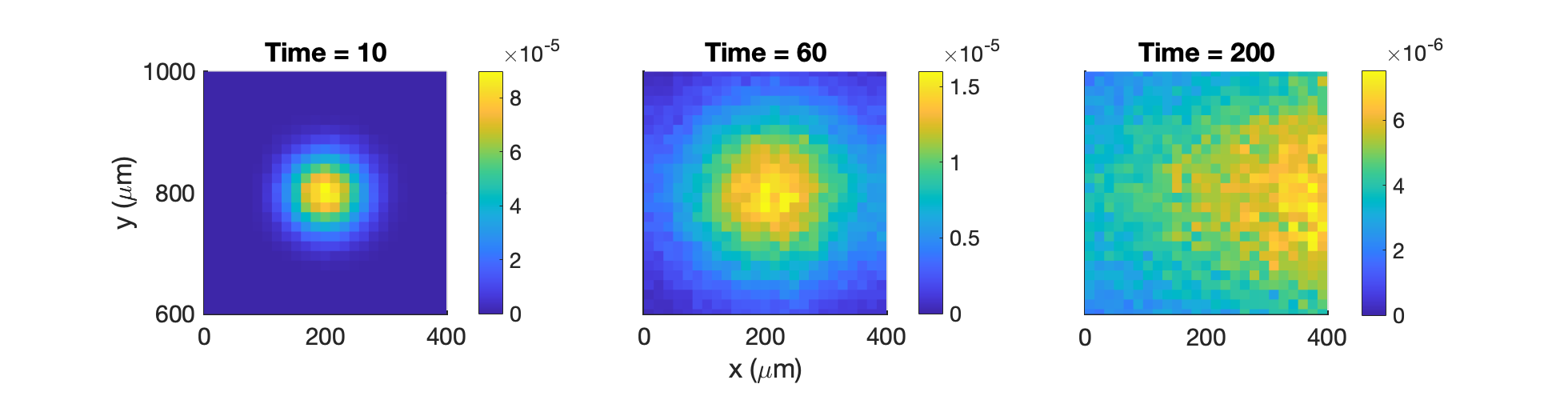}}\\
\vspace{-0.35cm}
\subfloat[Numerical solutions of  \eqref{main_equation_example_2d} for $\gamma=1.1$ and $p=0.1$]{\includegraphics[width=\textwidth]{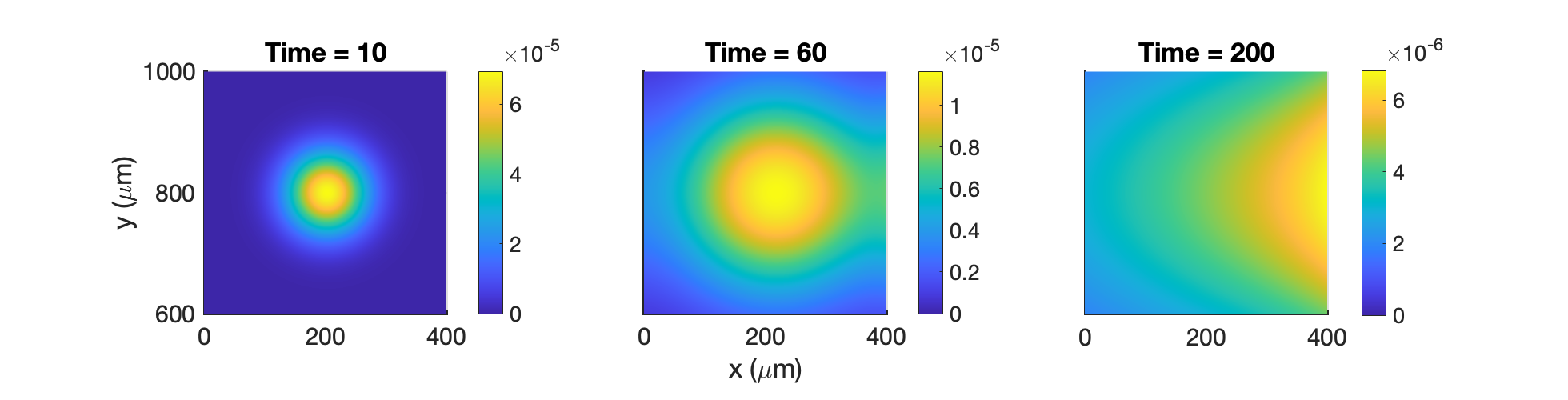}}\\
\vspace{-0.4cm}
\subfloat[$\text{CMC}_x$ (left) and $\text{CMC}_y$ (right)  ]{\includegraphics[width=0.8\textwidth]{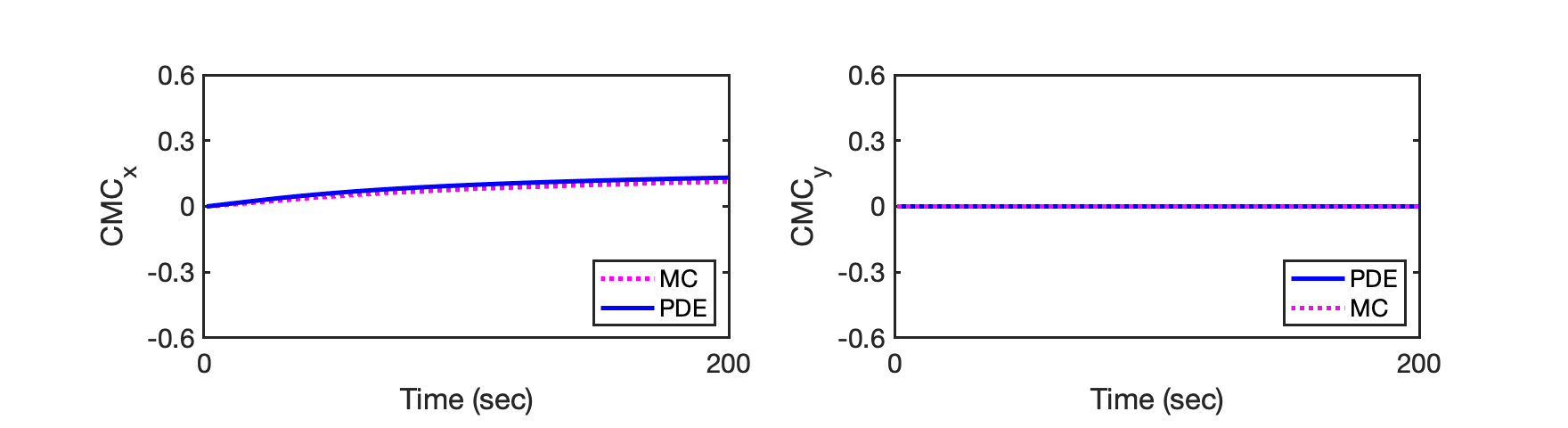}}
\caption{(a) and (b): Comparisons of Monte-Carlo simulation and numerical solutions of  \eqref{main_equation_example_2d} in response to two exponential gradients when $\gamma =1.1$. Plots in (a) and (b) are displayed only for $(x,y) \in [0, 400] \times [600, 1000]$. (c): Comparisons of the corresponding CMCs.}
\label{2d_exp_1p1}
\end{figure}

%******Fig 8******
\begin{figure}[ht!]
  \centering
\subfloat[Monte-Carlo simulation for $\gamma=0.9$ and $p=0.1$]{\includegraphics[width=\textwidth]{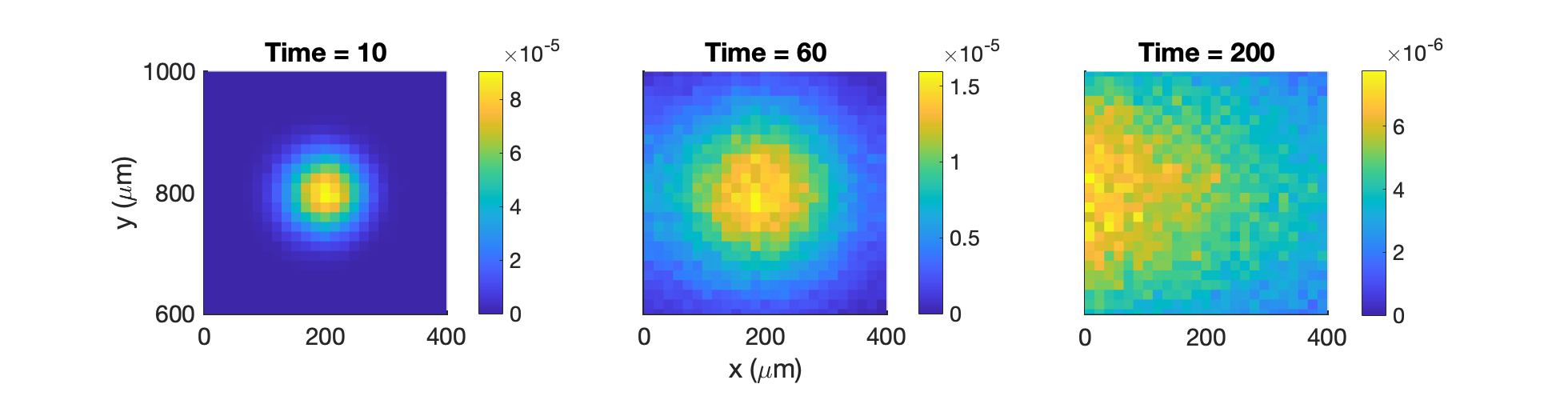}}\\
\vspace{-0.35cm}
\subfloat[Numerical solutions of  \eqref{main_equation_example_2d} for $\gamma=0.9$ and $p=0.1$]{\includegraphics[width=\textwidth]{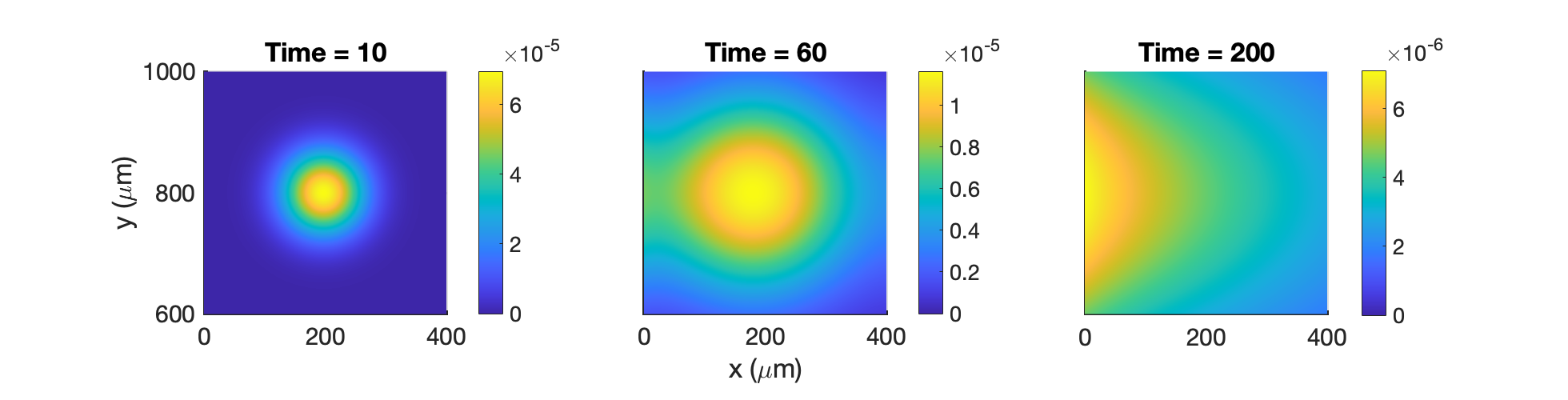}}\\
\vspace{-0.4cm}
\subfloat[$\text{CMC}_x$ (left) and $\text{CMC}_y$ (right)  ]{\includegraphics[width=0.8\textwidth]{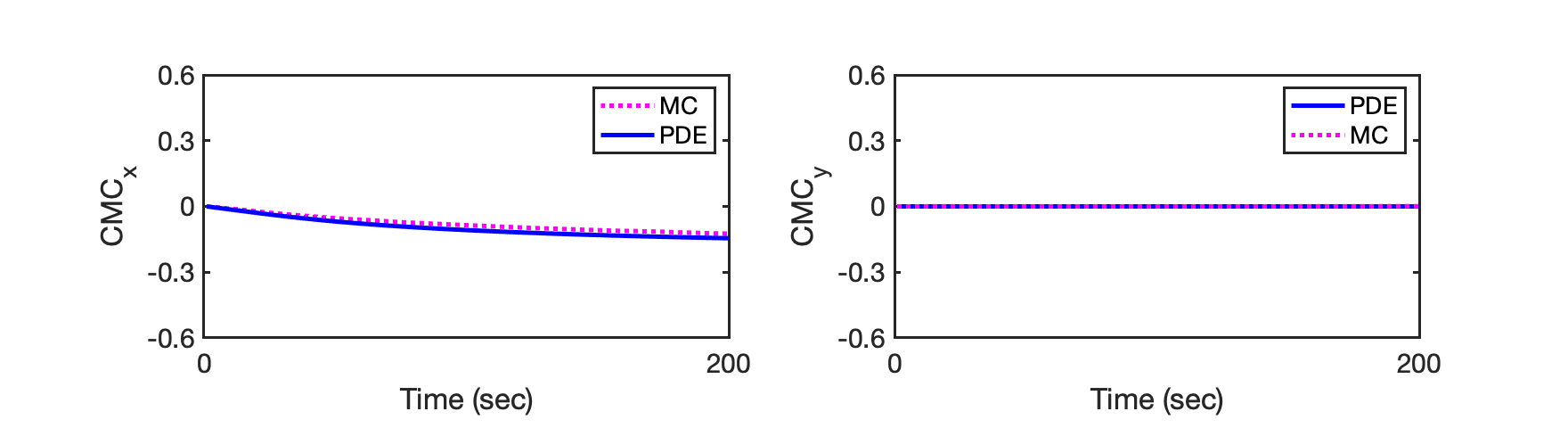}}
\caption{(a) and (b): Comparisons of Monte-Carlo simulation and numerical solutions of \eqref{main_equation_example_2d} in response to two exponential gradients when $\gamma =0.9$.  Plots in (a) and (b) are displayed only for $(x,y) \in [0, 400] \times [600, 1000]$. (c): Comparisons of the corresponding CMCs.}
\label{2d_exp_0p9}
\end{figure}

%*************
\subsection{Chemotaxis in response to mixed signals }\label{mix_ligands:2D}

In Sections \ref{lin_ligands:2D} and \ref{exp_ligands:2D}, we used two opposing gradients, independent of $y$, to reproduce chemotaxis experiments in the literature. In what follows, we assume that two opposing gradients MeAsp ($S_{1}$) and serine ($S_{2}$) satisfy
\begin{align}\label{mixed_gradients}
    S_{1}(x,y) = ( 0.5 x + 130)e^{0.005(y-800)} \quad \mbox{and} \quad S_{2}(x,y) = ( -0.03 x + 20)e^{-0.005(y-800)}.
\end{align}
Note that each gradient increases toward the corners $(0,0)$ and $(400, 1600)$ on the boundary of the domain, and reaches a peak at the corners. 
In this case, $V_1(x,y)= V(x)$, as defined in Section \ref{lin_ligands:1D}, and $V_2(x,y)=0.005\;\frac{\gamma-1}{\gamma+1}$. Therefore, condition \eqref{condition-V1-V2} holds. Furthermore, the bifurcation values are  $\gamma_1^*\approx 0.985$, as computed in  Section \ref{lin_ligands:1D}, and $\gamma_2^*=1$. Therefore, three scenarios occur:
(i) for $\gamma>1$ the bacteria move to the northeast, (ii) for $0.985<\gamma<1$ the bacteria move to northwest, and (iii) for $\gamma<0.985$ the bacteria move to southwest. 
As expected, the plots in Figure \ref{2d_Mix} show that bacteria accumulate toward the corner $(400,1600)$, when $\gamma =1.5>1$. Also, the solution of  \eqref{main_equation_example_2d} agrees well with the result of the Monte-Carlo simulation. 
%******Fig 9******
\begin{figure}[ht!]
  \centering
   \begin{minipage}{.5\linewidth}
\centering
\subfloat[Monte-Carlo simulation for $\gamma=1.5$ and $p=1$]{\includegraphics[width=\textwidth]{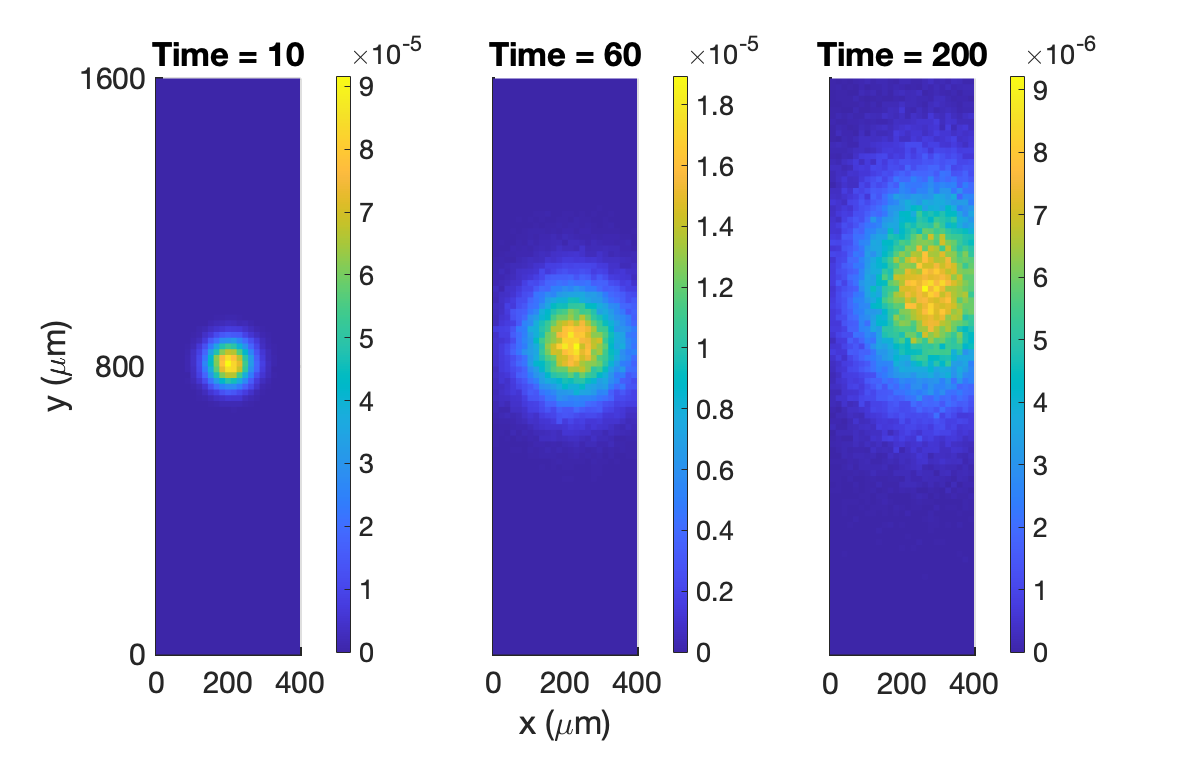}}
\end{minipage}%
\begin{minipage}{.5\linewidth}
\centering
\subfloat[Numerical solutions of  \eqref{main_equation_example} for $\gamma=1.5$ and $p=1$]{\includegraphics[width=\textwidth]{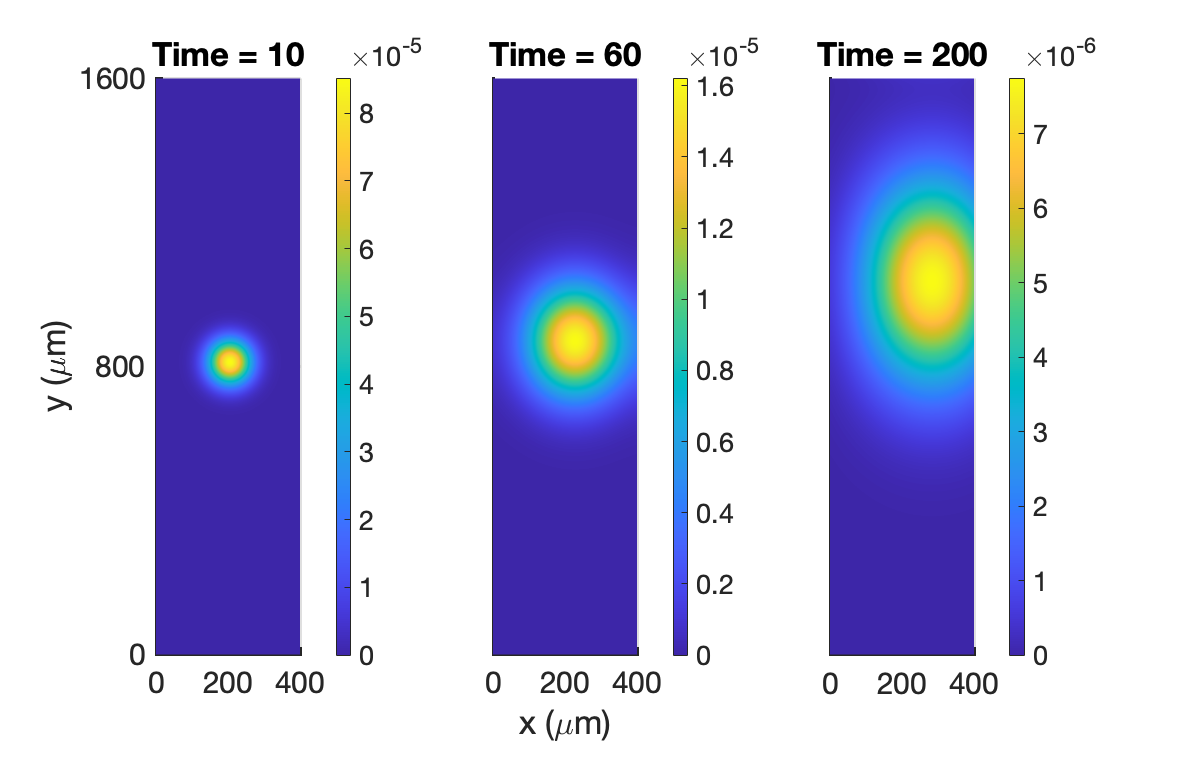}}
\end{minipage}\par
\vspace{-0.3cm}
\subfloat[$\text{CMC}_x$ (left) and $\text{CMC}_y$ (right)  ]{\includegraphics[width=0.8\textwidth]{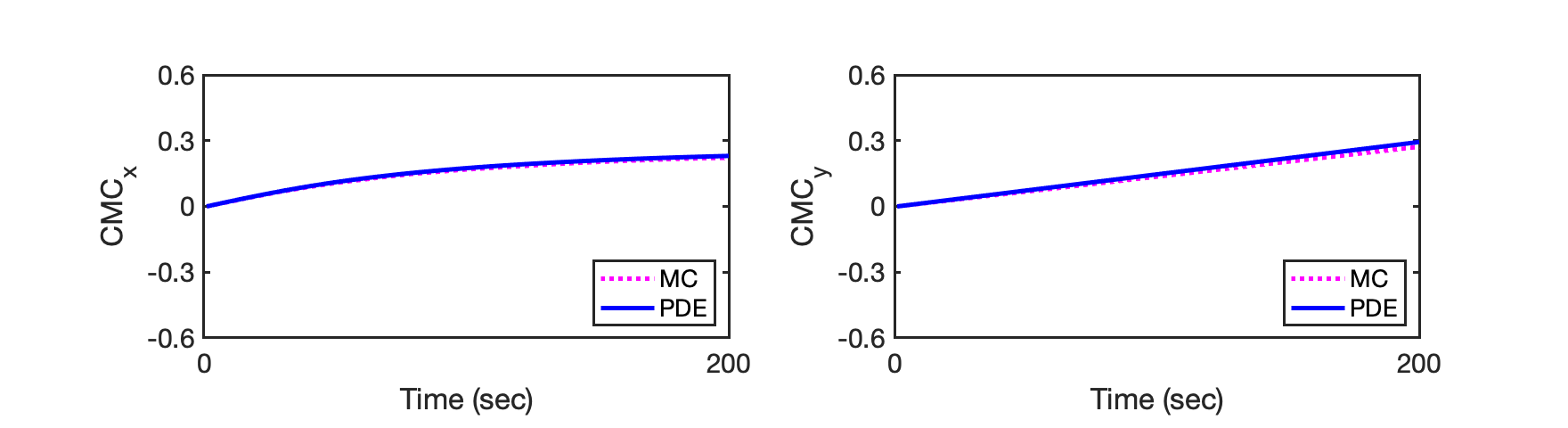}}
\caption{(a) and (b): Comparisons of Monte-Carlo simulation and numerical solutions of \eqref{main_equation_example_2d} for gradients \eqref{mixed_gradients} for $\gamma =1.5$. (c): Comparisons of the corresponding CMCs.}
\label{2d_Mix}
\end{figure}

%%%Discussion

\section{Discussion}\label{Discussion}
In this work, we studied the movement of a population of \emph{E.~coli} bacteria in response to two stimuli in a one- and a two-dimensional environment. Experimental results \cite{kalinin2010responses} show that the bacterial chemotactic preference to serine and MeAsp  depends on the ratio of their  chemoreceptors, namely $\gamma=\text{Tar}/\text{Tsr}$. In a shallow-gradient regime, we analytically found a threshold $\gamma^*$ that determines the bacterial preference, i.e., if $\gamma>\gamma^*$,  the bacteria move toward the gradient of MeAsp, and if $\gamma<\gamma^*$,  the bacteria move toward the gradient of serine. We examined our results in an environment where one stimulus  is dominant  everywhere and observed that in such a situation, a bigger force $\gamma^*$ might be needed to change the preference of the bacteria.

We started with a microscopic model for a population of bacteria carrying a one-dimensional internal dynamics. Indeed the microscopic equation is the forward Fokker-Planck equation of a stochastic model which describes bacterial chemotaxis \cite{stroock1974some}. Then, we approximated the microscopic Fokker-Planck equation by a macroscopic advection-diffusion equation which is more tractable mathematically. We compared the numerical solution of the advection-diffusion equation with a Monte-Carlo simulation of the bacterial chemotaxis to validate the approximation in
a shallow-gradient regime.

{The analysis in deriving the advection-diffusion equations is valid under the shallow-gradient condition. However, we numerically observed that even if the shallow-gradient condition does not hold, some of our results remain valid. For example, Figure \ref{fig:CMC} shows that under the condition of Section \ref{lin_ligands:1D}, the behavior of the bacteria does not change even when the adaptation rate $p$ does not satisfy the sallow-gradient condition (gray region). We also observed that $p$ does not affect the preference of bacteria. In fact, cells are often exposed to rapidly changing signals in vitro experiments and natural environments (see \cite{xue2015macroscopic,xue2016moment} and references therein), and great progress has been made in relaxing shallow gradient assumption 
\cite{xue2009multiscale, xue2015macroscopic, xue2016moment, rousset2013individual, gosztolai2020cellular}. Our work can be improved by considering a more general class of stimuli.
}

In \cite{long2017cell}, the authors found that \emph{E.~coli} cells respond to the gradient of chemoattractant not only by biasing their own random-walk swimming pattern through the intracellular pathway, but also by actively secreting a chemical signal into
the extracellular medium, possibly through a communication signal transduction pathway. 
The extracellular signaling molecule is a strong chemoattractant that attracts distant cells to the food source. They showed that 
cell-cell communication enhances bacterial chemotaxis toward external attractants. 
Incorporating such chemoattractant into microscopic model is one of the main areas of our future investigation. This cell-cell communication can be modeled as an external force to each cell and described by an extra term into the LHS of \eqref{transport}, see \cite{xue2009multiscale}.  

In this work, we only considered a one-dimensional internal dynamics. 
To obtain the internal dynamics of \emph{E.~coli} in response to multiple stimuli, we applied the heterogeneous MWC model \eqref{output:a} \cite{hu2014behaviors,mello2005allosteric,keymer2006chemosensing}, which can capture the total activity level of bacterium affected by the stimuli and  mathematically is tractable. 
In this model, all receptors within the cluster are assumed to turn on and off simultaneously, and therefore, only  the total kinase activity and total methylation level are considered.   However, in a mixed-receptor cluster, it was found that  receptor methylation dynamics is ligand specific. Hence, a local adaptation model, such as the Ising-type model, can better explain the adaptation dynamics of the mixed-receptor cluster, see e.g., \cite{keymer2006chemosensing} and \cite{hu2013precision}. 
Such models require higher dimensional equations to describe the internal dynamics. 
In our future works, we generalize our result to two-dimensional internal dynamics and for each receptor Tar and Tsr, we will consider  separate activity levels $\aa_1$ and $\aa_2$ instead of $\aa$ in \eqref{output:a} and separate methylation dynamics $dm_1/dt$ and $dm_2/dt$ instead of \eqref{methylation:m}.

%******Fig 10******
\begin{figure}[ht]
\floatbox[{\capbeside\thisfloatsetup{capbesideposition={right,top},capbesidewidth=6cm}}]{figure}[\FBwidth]
{\includegraphics[width=8cm]{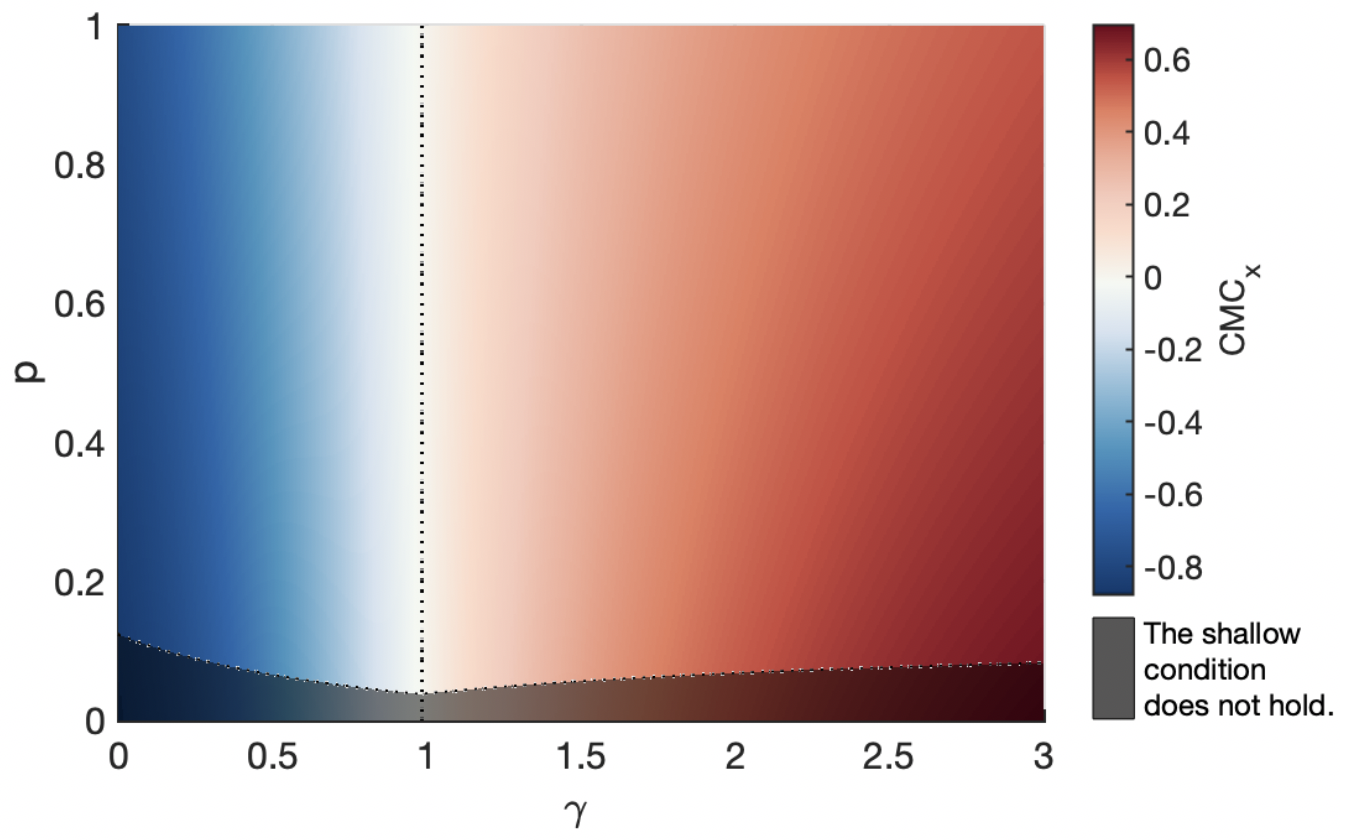}}
{\hspace{0.2cm}
\caption{\small{
How fast the signal changes or the adaptation speed does not affect the bacterial chemotactic preference.
${\rm CMC}_{x}$ of the steady state (\ref{steady_state}) 
for $S_{1}$ and $S_{2}$ in (\ref{linear_linear_gradients}) for $x_{0}=200$. 
For $(\gamma, p)$ in the dark red (respectively, blue) region, ${\rm CMC}_{x}$ becomes positive (respectively, negative)  as shown in the color bar. For $(\gamma,p)$ in the dark grey region, the shallow condition (\ref{Example:shallow:condition:1D}) is not satisfied. The dotted line represents $\gamma\approx 0.985$. }
}\label{fig:CMC}}
\end{figure}

% ******Acknowledgement******
\section{Acknowledgement }\label{Acknowledgement}
The authors would like to thank Professor Eduardo Sontag for sharing the Matlab codes for one-dimensional space (used in \cite{aminzare2013remarks}) and Professor Hans Othmer for helpful discussions. This work is partially supported by the  University of Iowa Old Gold Fellowship and Simons Foundation (712522) to ZA.

% ****** References******

%****** Appendix******
\appendix
\section{Appendix}\label{Appendix}

\subsection{Proof of Lemma \ref{lem:existence 1D}}\label{Appendix0}

The equation of our interest is
\begin{align*}
    n_{t} = \big( D n_{x} - \chi V(x) n\big)_{x}  = D n_{xx} - \chi V(x) n_{x} - \chi V'(x) n
\end{align*}
with boundary conditions
\begin{align*}
   D n_{x}(0,t) = \chi V(0) n(0,t), \qquad D n_{x}(L,t) =  \chi V(L) n(L,t).
\end{align*}
Assume $n(x,t)= \varphi(x) \psi(t)$. Then, it is satisfied
\begin{align*}
    \frac{\psi'(t)}{\psi(t)} = \frac{D \varphi''(x) - \chi V(x) \varphi'(x) - \chi V'(x) \varphi(x)}{\varphi(x)} =: - \lambda.
\end{align*}
To show that the solution $n(x,t)$ is bounded, we prove that if $\lambda$ exists, it is non-negative. 

Consider the following eigenvalue problem:
\begin{align}\label{EP}
 \mathcal{L}\varphi(x) :=   D\varphi''(x) - \chi V(x) \varphi'(x) - \chi V'(x) \varphi(x) = - \lambda \varphi(x)\tag{EP}
\end{align}
satisfying
\begin{align}\label{BC}
    D \varphi'(0)= \chi V(0) \varphi(0), \qquad D \varphi'(L) = \chi V(L) \varphi(L) .\tag{BC}
\end{align}
Putting (\ref{EP}) into the Sturm-Liouville operator, we have
\begin{align}\label{SL}
\mathcal{L}_{p(x)}  : = \frac{d}{dx} \Big( p(x) \frac{d}{dx}\Big) +q(x) 
 = - \lambda \sigma(x) 
,\tag{SL}
\end{align}
where 
\begin{align*}
    p(x) = e^{-\int \frac{\chi V(x)}{D} dx} >0, \qquad 
    q(x) = -\frac{\chi V'(x)}{D} p(x) \geq 0, \qquad \sigma (x) = \frac{1}{D} p(x) >0. 
\end{align*}

By Sturm-Liouville's Theory, 
 the problem (\ref{SL})-(\ref{BC}) is naturally posed on $H_{\rm bc}^{2},$ where
\begin{align*}
    H_{\rm bc}^{2}([0,L]) = \big\{ u \in H^{2}([0,L])  \;\ : \;\  D u_{x}(0,t) = \chi V(0) u(0,t), \quad D u_{x}(L,t) =  \chi V(L) u(L,t)
    \big\},
\end{align*}
and $\mathcal{L}_{p(x)}$ is self-adjoint
in the inner product
\[
<u, v> : = \int_{0}^{L} u(x) \overline{v(x)} dx.
\]
Moreover, the eigenvalues and the corresponding normalized eigenfunctions of (\ref{SL})-(\ref{BC}) satisfy the following properties:
\begin{enumerate}
    \item[(a)] All the eigenvalues are real, simple, and satisfy
   $ \lambda_{0} < \lambda_{1} < \lambda_{2} <\cdots$ and $\lim \limits_{n \rightarrow \infty}\lambda_{n}  = \infty. $
    \item[(b)] Each eigenfunction $\varphi_{n}(x)$ has $n$ simple zeros in the open interval $(0, L).$
    \item[(c)] $<\varphi_{n}, \varphi_{m}> = \delta_{nm}$.
    \item[(d)] $\{ \varphi_{n}(x)\}_{n=0}^{\infty}$ forms a complete orthonormal basis of $L^{2}([0,L])$. 
    \item[(e)] The smallest eigenvalue $\lambda_{0}$ is non-negative 
and satisfies
\begin{align*}
        &\qquad\mathcal{L}_{p(x)} \varphi_{0}(x) = - \lambda_{0} \sigma(x)  \varphi_{0}(x),\\
           &  \Rightarrow\; <\mathcal{L}_{p(x)} \varphi_{0}(x), \varphi_{0}(x)> \;=\; \underbrace{\int_{0}^{L} \frac{d}{dx} \Big( p(x) \frac{d}{dx} \varphi_{0}(x) \Big) \varphi_{0}(x) }_{=:\mathcal{I}_{1}}
        +
        \underbrace{q(x) \varphi_0^{2}(x) dx}_{=:\mathcal{I}_{2}}\\
        & \qquad \qquad \qquad \qquad  \;\ \quad \qquad = - \int_{0}^{L} \lambda_{0}  \sigma(x) \varphi_{0}^{2}(x) dx .
    \tag{$*$} \label{star}
    \end{align*}
    For simplicity we 
    replace $\varphi_{0}(x)$ and $\frac{d}{dx}$ by $u(x)$ and $'$, respectively. Then, by integration by parts, we have 
    \begin{align*}
    \mathcal{I}_{1} 
    = p(x) u(x) u'(x)\Big|_{0}^{L} - \int_{0}^{L} p(x)(u'(x))^{2}  dx, 
    \end{align*}
    and
    \begin{align*}
        \mathcal{I}_{2} 
    &= -\frac{\chi}{D}
    \int_{0}^{L} V'(x) e^{- \int \frac{\chi V(x)}{D} dx} u^{2}(x) \\
    &= - \frac{\chi}{D} V(x) e^{-\int \frac{\chi V(x)}{D} dx } u^{2}(x) dx \Big|_{0}^{L} 
    + \frac{\chi}{D}\int_{0}^{L} V(x) \Big(-\frac{\chi V(x)}{D} \Big) e^{-\int \frac{\chi V(x)}{D} dx} u^{2}(x) dx \\
    & \qquad + \frac{\chi }{D} \int_{0}^{L} V(x) e^{-\int \frac{\chi V(x)}{D} dx} 2 u(x) u'(x)dx\\
    & = - \frac{\chi}{D} V(x) p(x) u^{2}(x)  \Big|_{0}^{L} -  \int_{0}^{L} \frac{\chi^{2} V^{2}(x)}{D^{2}} p(x) u^{2}(x) dx + \int_{0}^{L} \frac{\chi V(x)}{D} p(x) 2 u(x) u'(x) dx.
    \end{align*}
    Note that $\eqref{star} = \mathcal{I}_{1} + \mathcal{I}_{2}$. Hence, 
    \begin{align*}
       \eqref{star} 
        &= p(L)u(L) \Big( u'(L) - \frac{\chi}{D} V(L) u(L)\Big) - p(0) u(0)\Big( u'(0) - \frac{\chi}{D} V(0) u(0)\Big)\\
        & \qquad + \int_{0}^{L}p(x) \Big( -u'^{2}(x)  -\frac{\chi^{2} V^{2}(x)}{D^{2}} u^{2}(x) +2 \frac{\chi V(x)}{D} u(x) u'(x) \Big)dx,
    \end{align*}
    where the first two terms on the right hand side disappear due to (\ref{BC}), and
    the integrand of the integral is non-positive since $p(x) >0$ and 
    \begin{align*}
        u'^{2} - 2 \frac{\chi V}{D} u u' + \frac{\chi^2 V^{2}}{D^{2}} u^{2} = \Big( u' - \frac{\chi V}{ D} u \Big)^{2} \geq 0.
    \end{align*}
    Therefore, from (\ref{star}), we arrive at
    \begin{align*}
          \lambda_{0} = -\dfrac{< \mathcal{L}_{p(x)} \varphi_{0},\varphi_{0}>}{\int_{0}^{L}  \sigma(x) \varphi_{0}^{2}(x) dx } \geq 0.
    \end{align*}
    
\end{enumerate}

\begin{landscape}
\subsection{Parameters in \emph{E. coli} internal dynamics}\label{Appendix1}

\begin{table}[ht!]
\centering
\captionof{table}{Parameters used in intracellular signaling pathway of \emph{E. coli}}
\label{Table:Parameters}
\scalebox{0.8}{
\begin{threeparttable}
\begin{tabular}{lclll}
	\toprule[0.5mm]
\multicolumn{1}{c}{ {\bf Equation}} 
& {\bf Parameter}&
\multicolumn{1}{c}{{\bf Description}}
& \multicolumn{1}{c}{{\bf Value}}
& \multicolumn{1}{c}{{\bf References}}
\\
\midrule[0.3mm]
\multirow{9}{*}{MWC model \eqref{output:a}}     
& $N$   & Number of receptors in a cluster, composed of Tar and Tsr & 6 & \cite{jiang2010quantitative,mello2007effects}
\\
& $r_{1}$    & Fraction of receptor Tar to MeAsp\tnote{a}   &      &  
\\ 
& $r_{2}$    & Fraction of receptor Tsr to serine\tnote{a}   &      & 
\\
& $\alpha$  & Free energy per added methylation group   & 1.7 & \cite{jiang2010quantitative,sourjik2002receptor, shimizu2006monitoring, endres2006precise}
\\
& $m_{0}$   & Reference methylation level in the free energy  & 1 &  \cite{jiang2010quantitative,sourjik2002receptor, shimizu2006monitoring, endres2006precise}     
\\
& $K_{A}^{1}$   & Dissociation constant of MeAsp to the active receptor Tar   & 18.2 $\mu M$    & \cite{kalinin2009logarithmic,jiang2010quantitative,mello2007effects,lan2011adapt}
\\
& $K_{A}^{2}$   & Dissociation constant of serine to the active receptor Tsr & 3 $mM$    & \cite{kalinin2009logarithmic, mello2007effects, lan2011adapt}
\\
& $K_{I}^{1}$   & Dissociation constant of MeAsp to the inactive receptor Tar   & 6 $\mu M$ &    \cite{kalinin2009logarithmic,jiang2010quantitative,lan2011adapt}           
\\
& $K_{I}^{2}$   & Dissociation constant of serine to the active receptor Tsr & 30 $\mu M$     &  \cite{kalinin2009logarithmic,lan2011adapt}
\\
\hline
\multirow{2}{*}{Adaptation model \eqref{methylation:m}}      
& $a_{0}$   & Adaptation level  & 0.5   & \cite{jiang2010quantitative,berg1975transient}
\\
& $\tau_{a}$    & Adaptation time    &  varies\tnote{b}  & \cite{vladimirov2009chemotaxis}
\\ 
\hline
\multirow{4}{*}{Run and Tumble motion \eqref{tumble_rate}} 
& $\lambda_{0}$ & Rotational diffusion   & 0.28 $rad^{2}s^{-1}$ & \cite{sourjik2002receptor, hu2013precision}        
\\
& $H$   & Hill coefficient of motor's response curve    & 10  &  \cite{jiang2010quantitative,sourjik2002receptor}                          
\\
& $\tau$    & Run average time  & 0.8 s    &   \cite{jiang2010quantitative,sourjik2002receptor}   
\\
& $\nu$ & Run velocity  & 16.5 $\mu m s^{-1}$   &   \cite{jiang2010quantitative,sourjik2002receptor}    
\\ \hline
\multirow{3}{*}{Transformed internal dynamics  \eqref{Example:translation}}      
& $q$   & $a_{0}$  & 0.5   & \cite{kalinin2009logarithmic}
\\
& $p$    & $\alpha \tau_{a}^{-1}$   &  varies\tnote{b} & 
\\ 
& $\gamma$    & Ratio between Tar and Tsr receptors, $r_{1} r_{2}^{-1}$    & varies\tnote{a}  & 
\\
\hline
\multirow{1}{*}{Transformed tumbling rate   \eqref{Exampletumbling:hat:2D}}      
& $r$   & $(\tau a_{0}^{H})^{-1}$  & 1280  & \cite{kalinin2009logarithmic}
\\
\bottomrule[0.5mm]
\end{tabular}
\begin{tablenotes}
\item[a] In this work, we are interested in the ratio of $r_{1}$ and $r_{2}$ satisfying $r_{1} + r_{2}=1$. Instead of the range of $r_{1}$ and $r_{2}$, we present the range of $\gamma=r_1/r_2$.

\item[b] 
Bacterial adaptation time varies, and it depends on the strength of signals \cite{vladimirov2009chemotaxis}.
In this work, we vary $p$ by choosing different $\tau_{\aa}$ between $1.7$ and $34$ as in  \cite{bray1995computer,terwilliger1986kinetics,simms1987purification}.
\end{tablenotes}
\end{threeparttable}
}
\end{table}

\newpage

%*************
\subsection{An overview of numerical simulations}\label{Appendix2}
\small{\textbf{A brief description of Monte-Carlo simulation:}
In a one-dimensional (respectively, two-dimensional) channel, we locate an ensemble of 100,000 agents in the center of the channel $x=200$ (respectively, $(x,y) = (200, 800)$) at time $t=0$. At each time step, the individuals choose a direction +1 or -1 (respectively, $(\cos (\theta), \sin (\theta))$, $\theta \in [0, 2\pi)$) at random, and move in that direction with a constant speed $\nu >0$. At each time step, the internal dynamics of each individual are 
computed by Euler method. At the end of each time step, we choose a number between 0 and 1 randomly
and compare the number with the probability of change from run to tumble in interval of length $d t$, namely $\lambda (\aa) dt$. If the turn occurs, the cell moves in the opposite direction with a probability of $0.5$ (respectively, rotates by $\theta \in [0, 2\pi)$, where $\theta$ is chosen at random). If a cell is located outside the spatial domain, we relocate the cell by imposing reflecting boundary conditions.
}
\normalsize

\begin{table}[ht!]
\centering
\captionof{table}{Input data used in numerical simulation}
\label{Table:Simulation}
\scalebox{0.7}{
\begin{threeparttable}
\begin{tabular}{lll}
\toprule[0.5mm]
\multicolumn{1}{c}{\multirow{2}{*}{\bf{Expression}}}                                           & \multicolumn{2}{c}{\bf{Value}}  \\
\cline{2-3} 
\multicolumn{1}{c}{} & \multicolumn{1}{c|}{\qquad \qquad \qquad \qquad \quad \bf{Monte-Carlo simulation} \qquad \qquad \qquad \qquad \quad } 
& \multicolumn{1}{c}{\bf{Numerical partial differential equations}\tnote{a}}
\\ 
\midrule[0.3mm]
Spatial domain ($\bf{x}$) 
& \multicolumn{2}{l}{\begin{tabular}[c]{@{}l@{}}1D: $0 \leq x \leq 400 $ ($\mu m$) \\ 2D: $0 \leq x \leq 400 $ ($\mu m$) and $0 \leq y \leq 1600$ ($\mu m$)\end{tabular}} \\ \hline
Time domain ($t$)  & \multicolumn{2}{l}{$0 \leq t \leq 200$ (sec)\tnote{b} } 
\\ \hline
Initial data   & \multicolumn{1}{l|}{An ensemble of 100,000 agents poses in the center of each domain. }   
&
\begin{tabular}[c]{@{}l@{}}
1D: $\exp\{- (x-200)^{2}/(2\varepsilon^{2})\}/\sqrt{2 \pi} \varepsilon, \quad \varepsilon=10^{-6}$
\\
2D: $\exp\{- ((x-200)^{2}+(y-800)^{2})/(2\varepsilon^{2})\}/2 \pi \varepsilon^{2}, \quad \varepsilon=10^{-6}$.
\end{tabular}          \\ \hline
\\ \hline
Boundary conditions
& \multicolumn{1}{l|}{No flux boundary conditions }                           & \begin{tabular}[c]{@{}l@{}}
No flux boundary conditions: \eqref{Robin_condition} for 1D and \eqref{Robin_condition_2D} for 2D
\end{tabular}          \\ \hline
Spatial step size ($\Delta x, \Delta y$)
& \multicolumn{1}{l|}{\begin{tabular}[c]{@{}l@{}}1D: $\Delta x = 0.01$ ($\mu m$)\\ 2D: $\Delta x=\Delta y=0.001$ ($\mu m $)\end{tabular}}
& \begin{tabular}[c]{@{}l@{}}1D: $\Delta x = 0.8$ ($\mu m$)\tnote{c}\\ 2D: $\Delta x=\Delta y=0.8$ ($\mu m)$\tnote{c} \end{tabular}                                \\ \hline
Time step size ($\Delta t$)                                                              & \multicolumn{1}{l|}{$\Delta t=0.0001$ (sec)\tnote{d}}                  & 
\begin{tabular}[c]{@{}l@{}}
1D: $\Delta t = 0.001$ (sec) 
\\
2D: $\Delta t = 0.0008$ for Sections \ref{lin_ligands:2D}, \ref{exp_ligands:2D} and $\Delta t = 0.00025$ for Section \ref{mix_ligands:2D}\tnote{c} 
\end{tabular}     
                   \\ \hline
Ligand function ($S_{1}, S_{2}$) 
& \multicolumn{2}{l}{\begin{tabular}[c]{@{}l@{}}
$\cdot$ Dual linear gradients in Sections \ref{lin_ligands:1D} and \ref{lin_ligands:2D}: $S_{1}(x) = 0.5 x + 130$ and $S_{2}(x) = -0.03 x + 20$ 
\\
$\cdot$ Dual exponential gradients in Sections \ref{exp_ligands:1D} and \ref{exp_ligands:2D}: $S_{1}(x) = 130 e^{0.0023 x}$ and $S_{2}(x) = 8 e^{-0.0023 (x-400)}$ 
\\
$\cdot$ Mixed opposing gradients in Section \ref{mix_ligands:2D}: 
$S_{1}(x,y)=(0.5 x + 130 ) e^{0.005(y-800)}$ and $S_{2}(x,y)=(-0.03x + 20 )e^{-0.005(y-800)}$ 
\vspace{0.1cm}
\end{tabular}}                                     \\ \hline
\begin{tabular}[c]{@{}l@{}}Tar/Tsr ratio ($\gamma$)\\ Adaptation\tnote{e} \; ($p$) \end{tabular} & \multicolumn{2}{l}{\begin{tabular}[c]{@{}l@{}}
$\cdot$ Sections \ref{lin_ligands:1D} and \ref{lin_ligands:2D}:
$\gamma=1.5$ and $\gamma = 0.5$ with $p=0.4$ (1D), 1 (2D).
\\ 
$\cdot$ Sections \ref{exp_ligands:1D} and \ref{exp_ligands:2D}:
$\gamma=1.1$ and $\gamma=0.9$ with $p=0.05$ (1D), 0.1 (2D)\\
$\cdot$ Section \ref{mix_ligands:2D}: $\gamma=1.5$ and $p=1$
\end{tabular}}
\\
\bottomrule[0.5mm]
\end{tabular}
\begin{tablenotes}
\item[a] {Finite difference method is used.}
\item[b] 
The solutions of the Monte-Carlo simulation and advection-diffusion equations shown in Sections \ref{lin_ligands:1D} to \ref{exp_ligands:2D} become stationary at $t=200.$ 
\item[c] {The time step size and the space step size in the finite difference formula
satisfy
Courant-Friedrichs-Lewy (CFL) condition and von Neumann stability analysis, respectively.
It is confirmed that using smaller step sizes does not affect our results as long as the stability conditions are satisfied.
}
\item[d] {We choose small value for $\Delta t$ to solve the dynamics of methylation by Euler method.}
\item[e] {The values of $p$ and $\gamma$ satisfy the shallow condition.}
\end{tablenotes}
\end{threeparttable}
}
\end{table}

\end{landscape}

\end{document}